\documentclass[11pt, oneside]{article} 
\usepackage[margin=2cm]{geometry}                		
\usepackage[parfill]{parskip}    		
\usepackage{graphicx}			
\usepackage[utf8]{inputenc}
\usepackage[english]{babel} % English language/hyphenation
\usepackage{amsmath,amssymb,amsthm,bm} % Math packages
\usepackage[margin=2cm]{geometry}
\usepackage[color=yellow]{todonotes}
\usepackage{mathrsfs}
\usepackage{url}	
\usepackage{esint}
\usepackage[colorlinks]{hyperref}
\usepackage{enumerate}
\usepackage{outlines}[enumerate]
\usepackage{framed,comment,enumerate}
\usepackage{upgreek}
\usepackage{pgfplots}
\usepackage{float}
\usepackage{tikz,tikz-cd}
\usepackage{rotating}
\usepackage{graphicx}
\usepackage{caption}
\usepackage{multicol}
\usepackage{cancel}
\usepackage{subcaption}
\usepackage[title]{appendix}
\usepackage{tikz-cd}
\usepackage{pgfgantt}
\usepackage{rotating}
\usepackage{natbib}
\extrafloats{100}
\numberwithin{equation}{section}

\newtheorem{proposition}{Proposition}[section]
\newtheorem{remark}{Remark}[section]

\theoremstyle{definition}

\DeclareFontFamily{U}{MnSymbolC}{}
\DeclareSymbolFont{MnSyC}{U}{MnSymbolC}{m}{n}
\DeclareFontShape{U}{MnSymbolC}{m}{n}{
    <-6>  MnSymbolC5
   <6-7>  MnSymbolC6
   <7-8>  MnSymbolC7
   <8-9>  MnSymbolC8
   <9-10> MnSymbolC9
  <10-12> MnSymbolC10
  <12->   MnSymbolC12}{}
\DeclareMathSymbol{\intprod}{\mathbin}{MnSyC}{'270}

\newcommand{\wt}[1]{\widetilde{#1}}
\newcommand{\wh}[1]{\widehat{#1}}

\newcommand{\mc}[1]{\mathcal{#1}}
\newcommand{\bs}[1]{\boldsymbol{#1}}

\newcommand{\mcal}[1]{\mc{#1}}
\newcommand{\scp}[2]{\left<#1\,,\,#2\right>}

\newcommand{\ad}{\operatorname{ad}}

\DeclareMathOperator{\diff}{d\!}
\def\p{{\partial}}

\def\bB{{\mathbf{B}}}

\def\bm{{\mathbf{m}}}

\def\bR{{\mathbf{R}}}

\def\bu{{\mathbf{u}}}

\def\bx{{\mathbf{x}}}

\def\p{\partial}

\pgfplotsset{compat=1.16}

\def\bm{\boldsymbol{m}}
\def\bu{\boldsymbol{u}}

\def\bx{\boldsymbol{x}}
\def\bB{\boldsymbol{B}}
\def\bxi{\boldsymbol{\xi}}

\begin{document}

\title{\textbf{Plasma dynamics in thin domains}}
\author{Darryl D. Holm\thanks{Department of Mathematics, Imperial College London} , Ruiao Hu\footnotemark[1] , and Oliver D. Street\thanks{Grantham Institute, Imperial College London} \\ 
\footnotesize
d.holm@ic.ac.uk, ruiao.hu15@imperial.ac.uk, o.street18@imperial.ac.uk 
\\  \small
\small
Lie group invariant variational principles; Magnetohydrodynamics; shallow water equations
}
\date{\today}

\maketitle

\begin{abstract}
In the present work, we study the geometric structures of the Rotating Shallow Water Magnetohydrodynamics (RSW-MHD) equations through a Lie group invariant Euler--Poincar\'e variational principle. In this geometric framework, we derive new, structure-preserving stochastic RSW-MHD models by introducing stochastic perturbations to the Lie--Poisson structure of the deterministic RSW-MHD equations. The resulting stochastic RSW-MHD equations provide new capabilities for potential application to uncertainty quantification and data assimilation, for example, in space plasma (space weather) and solar physics, particularly in solar tachocline dynamics. 
\end{abstract}

\tableofcontents

\section{Introduction}\label{Intro-sec}

Astrophysical plasma dynamics often takes place in thin domains (i.e., domains of small aspect ratio) such as accretion disks, planetary atmospheres and transition zones in stars. Hence, a variational derivation of the equations and a discussion of the geometric properties of the dynamics of plasmas in thin domains may provide a useful setting for developing new capabilities in astrophysics. 
 
In particular, \citet{gilman2000magnetohydrodynamic} introduced equations for rotating shallow water magnetohydrodynamics (RSW-MHD) as a model of solar tachocline dynamics.  The solar tachocline is a relatively thin layer near the Sun's surface (about 1/20 of the solar radius) comprised of magnetised quasineutral plasma which bridges between the Sun's convective zone and its radiative zone \citep{hughes2007solar}. Since their derivation, the RSW-MHD equations have been investigated from several theoretical and modelling points of view. In particular, the Green--Naghdi approach for including dispersive effects was investigated in \cite{dellar2003dispersive}. Studies of linear and non-linear waves in RSW-MHD have been treated by \cite{schecter2001shallow} and investigations of shear-flow instabilities by \cite{mak2016shear}. The RSW-MHD system has also been shown to be hyperbolic and possess a Hamiltonian structure \citep{de2001hyperbolic, dellar2002hamiltonian, rossmanith2003constrained}. Extensions to multi-layer RSW-MHD equations have also been documented \citep{hunter2015waves, zeitlin2013remarks,alonso2021asymptotic}.

%\todo[inline]{OS: Introduction needs to motivate stochastic modelling.\\
%DH: See next blurb.}

The present paper has two main objectives: 
\begin{enumerate}[(i)]
\item 
Use Lagrangian reduction by Lie symmetry in Hamilton's variational principle to derive and investigate the geometric structure and solution properties of the deterministic RSW-MHD equations; 
\item 
Derive stochastic methods for quantifying uncertainty and assimilating data that preserve the geometric structure and solution properties found in (i) for the deterministic RSW-MHD equations. 
\end{enumerate}

\paragraph{Summary of the paper.} The contributions of the subsequent sections are as follows.
\begin{itemize}
\item 
    Section \ref{Intro-RSW-MHD-sec} introduces the Gilman model of Rotating Shallow Water MHD (RSW-MHD) and considers its Lagrangian and the Hamiltonian formulation.   
    \begin{enumerate}[(i)]
    \item 
    In Section \ref{RSW-sec}, we demonstrate the use of Lagrangian reduction by symmetry in Hamilton's principle for deriving the RSW-MHD equations. 
    \item In Section \ref{Ham-RSW-MHD-sec}, we transform the variational derivation of the RSW-MHD equations on the Lagrangian side to the Hamiltonian formulation and demonstrate its Lie-Poisson structure.   
    \end{enumerate}    

\item 
Section \ref{TRSW-MHD-sec} extends the RSW-MHD model to include thermal effects. This extension results in a Thermal Rotating Shallow Water MHD (TRSW-MHD) model and we consider its Lagrangian and Hamiltonian formulations.

\item Section \ref{Stoch-RSW-MHD} introduces the following two classes of structure-preserving stochasticity into the RSW-MHD model. 
    \begin{enumerate}[(i)]
    \item 
    Stochastic Advection by Lie Transport (SALT) \citep{holm2015variational} which preserves the Casimir invariances of the deterministic model and is presented in Section \ref{sec:SALT}. 
    \item 
    Stochastic Forcing by Lie Transport (SFLT) \citep{HH2021a} which preserves energy of the deterministic model and is presented in Section \ref{sec:SFLT}.
    \end{enumerate}
\item    Section \ref{Conc-Out-sec} summarises the new perspectives and outlook for applying the new stochastic formulations achieved in the present results. 
\end{itemize}

\section{Rotating Shallow Water MHD (RSW-MHD)}\label{Intro-RSW-MHD-sec}
The motion equation for Rotating Shallow Water MHD (RSW-MHD) in a thin Euclidean domain, with bathymetry ${h}(\bx)$ for $\bx\in \mathbb{R}^2$, may be written in Cartesian vector calculus notation in  \cite{gilman2000magnetohydrodynamic} as,
\begin{equation}\label{eqn:RSW-MHD-motion-VC}
    \p_t \bs{u} + \bs{u} \cdot \nabla \bs{u}+ f \bs{u}^\perp = -g \nabla(\eta - {h}(\bx)) + \bB \cdot \nabla \bB \,,
\end{equation}
where $\bu^\perp = (-u_2,u_1)$. The dynamical RSW-MHD variables in \eqref{eqn:RSW-MHD-motion-VC} are the horizontal velocity, $\bu=(u_1,u_2)$, horizontal magnetic field, $\bB=(B_1,B_2)$, and layer thickness, $\eta$. The parameters for Coriolis force and gravitational acceleration are, respectively, $2 \boldsymbol{\Omega}$ and $g$. The term $\bB \cdot \nabla \bB$ in equation \eqref{eqn:RSW-MHD-motion-VC} is the two-dimensional expression of the three-dimensional $J \times B$ force with current density $ J=:{\rm curl\,}B $. A schematic of the RSW-MHD system is shown in Figure \ref{fig:schematic}.
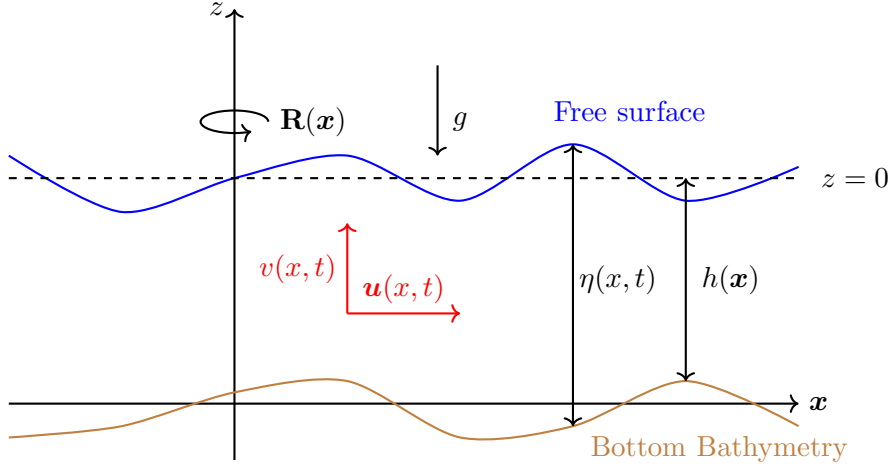
\begin{figure}
\centering
    \begin{tikzpicture}[scale=1.5, thick]

    % Draw coordinate axes
    \draw[->] (-2, 0) -- (5, 0) node[right] {$\bx$};
    \draw[->] (0, -0.5) -- (0, 3.5) node[left] {$z$};
    \draw[brown, thick] plot [smooth] coordinates {(-2,-0.3) (-1,-0.2) (0, 0.1) (1, 0.2) (2,-0.3) (3,-0.2) (4,0.2) (5, -0.2)};
    \node[brown] at (4.3, -0.4) {Bottom Bathymetry};
    
    % Water surface and free surface displacement
    \draw[blue, thick] plot [smooth] coordinates {(-2,2.2) (-1,1.7) (0,2.0) (1,2.2) (2,1.8) (3,2.3) (4,1.8) (5,2.1)};
    % \node[blue] at (3.5,2.6) {$\eta(x,t)$};
    \node[blue] at (3.5,2.6) {Free surface};
    
    % Still water level
    \draw[dashed] (-2,2) -- (5,2);
    \node at (5.5 ,2) {$z = 0$};
    
    % Depth label
    \draw[<->] (3,-0.2) -- (3,2.3);
    \node at (3.4,1.1) {$\eta(x,t)$};
    
    % Total water height h
    \draw[<->] (4,0.2) -- (4,2.0);
    \node at (4.4,1.1) {$h(\bx)$};
    
    % Arrows for velocity components
    \draw[->,red,thick] (1, 0.8) -- (2, 0.8) node[midway, above] {$\bu(x,t)$};
    \draw[->,red,thick] (1, 0.8) -- (1, 1.6) node[midway, left] {$v(x,t)$};
    
    % Rotation symbol
    \node at (0.7,2.5) {$\bR(\bx)$};
    \draw[->,black, thick] (0.3,2.5) arc[start angle=0, end angle=300, x radius=0.3, y radius = 0.1];

    % gravity symbol
    \node at (2,2.5) {$g$};
    \draw[->,black, thick] (1.8,3.0) -- (1.8, 2.2);
\end{tikzpicture}
    \caption{Schematic of the notation for RSW-MHD.}
    \label{fig:schematic}
\end{figure}

The non-dimensional form of the momentum equation \eqref{eqn:RSW-MHD-motion-VC} can be found by scaling with the following units. Let $L$ be the characteristic horizontal length scale, $H$ be the typical depth, $U$ be the characteristic velocity scale and $B_0$ the characteristic value of the magnetic field. Two different Rossby numbers that measure the relative strength of rotation can be found as $\mathrm{Ro} = U/(fL)$ and $\mathrm{Ro}_m = B_0/(fL)$, which are the `ordinary' Rossby number and the `magnetic' Rossby number, respectively \citep{zeitlin2013remarks, LZ2022}. Additionally, the gravitational effects are characteristic by the Froude number $\mathrm{Fr}^2 = U^2/(gH)$. Using these non-dimensional numbers, the non-dimensional form of \eqref{eqn:RSW-MHD-motion-VC} becomes
\begin{align}
    \p_t \bs{u} + \bs{u} \cdot \nabla \bs{u} + \frac{1}{\mathrm{Ro}}f\bs{u}^\perp = -\frac{1}{\mathrm{Fr}^2} \nabla(\eta - {h}(\bx)) + \frac{1}{\mu^2}\bB \cdot \nabla \bB \,. \label{eqn:RSW-MHD-motion-VC-nondim}
\end{align}
Here, $\mu := \mathrm{Ro}/\mathrm{Ro}_m$ is the ratio of the fluid and magnetic Rossby numbers. For balance between rotational and gravitational effects, one requires $\mathrm{Ro}/\mathrm{Fr}^2 = \mcal{O}(1)$.

Two advection relations hold as auxiliary equations for the RSW-MHD motion \eqref{eqn:RSW-MHD-motion-VC},
\begin{align}
    \p_t\eta +\nabla \cdot(\eta \bs{u}) &= 0
    \label{eqn:advection-eta-VC}
    \,,\\
    \p_t\bB + \bs{u} \cdot \nabla \bB -\bB \cdot \nabla \bs{u}&=0
    \label{eqn:advection-B-VC}
    \,.
\end{align}
Together, these auxiliary equations imply preservation of the condition
\begin{equation}\label{eqn:divergence-condition}
    \nabla \cdot(\eta \bB)=0 \,,
\end{equation}
which can therefore be regarded as a non-dynamical constraint on the initial values.

\subsection{A variational principle for the  RSW-MHD equations}\label{RSW-sec}
The advection equations \eqref{eqn:advection-eta-VC} and \eqref{eqn:advection-B-VC} may each be written in a coordinate free geometric form as
\begin{align}
    (\p_t + \mathcal{L}_u)(\eta\,d^2x) &= 0
    \label{eqn:advetion-eta}
    \,,\\
    (\p_t + \mathcal{L}_u)B &= 0
    \label{eqn:advection-B}
    \,,
\end{align}
where $\eta\,d^2x \in \Omega^2(M)$ is a volume form and $u,B \in \mathfrak{X}(M)$ are vector fields which may be expressed in the Cartesian coordinate system as $u = \bu\cdot\nabla$ and $B = \bB\cdot\nabla$. In this notation, the condition \eqref{eqn:divergence-condition} may be expressed as
\begin{equation}
    \mathcal{L}_B(\eta d^2x) = 0 \,.
\end{equation}
\begin{remark}\label{rmk:divergence-condition}
    By using the advection equations in their geometric form \eqref{eqn:advetion-eta} and \eqref{eqn:advection-B}, one may prove advection of the divergence $\operatorname{div}(\eta \bB)\,d^2x = \mathcal{L}_B(\eta\,d^2x)$. Indeed, Cartan's formula for the Lie derivative gives
    \begin{equation*}
        \mathcal{L}_{B}(\eta\,d^2x) = d\left( B \intprod (\eta\,d^2x) \right) + B \intprod d(\eta\,d^2x) = d\left( B \intprod (\eta\,d^2x) \right) \,,
    \end{equation*}
    since the exterior derivative of a volume form is zero. The exterior derivative commutes with both the time derivative and the Lie derivative. Consequently, the advection of $B$ and $\eta\,d^2x$ implies that
    \begin{equation*}
        (\p_t + \mathcal{L}_u)\mathcal{L}_{B}(\eta\,d^2x) = 0 \,.
    \end{equation*}
    Therefore, if the quantity $\operatorname{div}(\eta \bB)$ vanishes initially, then it will remain zero for all time.
\end{remark}
As shown by \citet{HMR1998}, these advected quantities break the symmetry of the Lagrangian in Hamilton's principle under the full diffeomorphism group which thereby leads to Lie-Poisson equations on the semidirect product between the co-algebra $\mathfrak{g}^* = \mathfrak{X}^*(M)$ and the space of advected quantities $ V^* = \mathfrak{X}(M)\oplus\Omega^2(M)$. To construct the semidirect product space $\mathfrak{g} \ltimes V$, we need a representation of the Lie group $G$, the diffeomorphism group $\textrm{Diff}(M)$, on $V$. The representation of the diffeomorphisms on $\Omega^2(M)$ is by pullback, and on $\mathfrak{X}(M)$ by the Lie group adjoint representation, $\mathrm{Ad}$. The corresponding infinitesimal action of $\mathfrak{X}(M)$ on $\Omega^2(M)$ is by Lie derivative and its action on $\mathfrak{X}(M)$ is by adjoint representation, $\mathrm{ad}$. The dynamical equations arising from Hamilton's principle with broken symmetry emerge as Euler-Poincar\'e equations with advected quantities for a Lagrangian $\ell:\mathfrak{X}(M)\times\Omega^2(M)\times\mathfrak{X}(M) \rightarrow \mathbb{R}$.
Indeed, the RSW-MHD equations \eqref{eqn:RSW-MHD-motion-VC-nondim}-\eqref{eqn:advection-B-VC} are the Euler-Poincar\'e equations corresponding to Hamilton's principle applied to the following action integral in Cartesian coordinates
\begin{equation}\label{eqn:RSW-MHD-action}
\begin{aligned}
    S &= \int_{0}^{T} \ell(u, \eta, B) \,dt 
    \\
    &= \int_{0}^{T} \int_{M}\left(\frac{1}{2}|\bs{u}|^{2}
    +\frac{1}{\mathrm{Ro}} \bs{u} \cdot \boldsymbol{R}(\boldsymbol{x})
    -\frac{1}{2\mu^2}|\bB|^{2}
    -\frac{1}{2 \mathrm{Fr}^{2}}(\eta-2 {h}(\boldsymbol{x}))\right) \eta  \,d^2x \, dt
    \,,
\end{aligned}
\end{equation}
where $M \in \mathbb{R}^{2}$ denotes the horizontal cross-section.
\begin{remark}[The coordinate free action]
    When the manifold is equipped with a suitable metric, the vector fields $u \in \mathfrak{X}(M)$ and $B \in \mathfrak{X}(M)$ possess associated $1$-forms, $u^{\flat} \in \Omega^1(M)$ and $B^{\flat} \in \Omega^1(M)$. In Euclidean domains this association is natural, since if $u = \bu\cdot\nabla$ then $u^\flat = \bu\cdot d\bx$ where $d\bx$ is the dual basis to $\nabla$. In this natural association, the coefficients of the basis elements do not change when using the `musical isomorphism' to transform between vector fields and $1$-forms. As such, the terms in the action can be simply expressed in a coordinate free exterior calculus notation as
    \begin{equation}
        |\bu|^2 = u \intprod u^\flat \,,\quad |\bB|^2 = B \intprod B^{\flat} \,,\quad\hbox{and}\quad \bu\cdot\bs{R} = u \intprod R \,,
    \end{equation}
    where $\intprod$ denotes the insertion of a vector field into a $1$-form and $R$ is a $1$-form such that $R = \bs{R}(\bx)\cdot d\bx$. This identification also assists in understanding the corresponding geometric spaces of the variational derivatives.
\end{remark}

\begin{remark}[The form of the Lagrangian]
    The Lagrangian in \eqref{eqn:RSW-MHD-action} is the standard Lagrangian for the rotating shallow water equations, augmented with a contribution from the energy density of the $\bB$-field. This is the construction applied by \cite{HMR1998} for an Euler equation coupled to a magnetic field.
\end{remark}

\begin{remark}[Weakly versus strongly magnetized regimes]
    The parameter $\mu$ controls the strength of the magnetic forces relative to the gravitational forces and the Coriolis force. When $\mathrm{Ro}/\mathrm{Fr}^2 = \mcal{O}(1)$, the dynamics are in geostrophic balance. Additionally, when $\mathrm{Ro}/\mu^2 = \mcal{O}(1)$, the dynamics are in a strongly magnetized regime where the balance with the Coriolis force involves both the hydrostatic pressure gradient force and the magnetic $J \times B$ force. In the weakly magnetized regime when $\mathrm{Ro}/\mu^2 \ll 1$ the balance with the Coriolis force involves only the the hydrostatic pressure gradient force.
\end{remark}

The RSW-MHD equations \eqref{eqn:RSW-MHD-motion-VC}-\eqref{eqn:advection-B-VC} will be derived by first evaluating the variational derivatives for the Lagrangian in the action integral \eqref{eqn:RSW-MHD-action}. Namely,

\begin{equation}\label{eqn:RSW-MHD-vars}
\begin{aligned}
\frac{1}{\eta} \frac{\delta \ell}{\delta u} & =\left(\bs{u}+\frac{1}{\mathrm{Ro}} \boldsymbol{R}(\boldsymbol{x})\right) \cdot \mathbf{d} \boldsymbol{x}=: \boldsymbol{V}(\boldsymbol{x}, t) \cdot \mathrm{d} \boldsymbol{x}=: V^{\flat} \in \Omega^1(M) 
\,,\\
\frac{\delta \ell}{\delta \eta} & =\left(\frac{1}{2}|\bs{u}|^{2}+\frac{1}{\mathrm{Ro}} \bs{u} \cdot \boldsymbol{R}(\boldsymbol{x})-\frac{1}{2\mu^2}|\bB|^{2}-\frac{1}{\mathrm{Fr}^{2}}(\eta-{h}(\boldsymbol{x}))\right)=: \beta(\boldsymbol{x}, t) \in \Omega^0(M) 
\,,\\
\frac{\delta \ell}{\delta B} & = -\frac{1}{\mu^2}\bB \cdot \mathrm{d} \boldsymbol{x} \otimes \eta d^{2} x=: \frac{1}{\mu^2}B^{\flat} \otimes \eta\, d^{2}x \in \mathfrak{X}^*(M)
\,,
\end{aligned}
\end{equation}
where $\mathfrak{X}^*(M)$ denotes the dual space to the space of vector fields, which contains $1$-form densities. In order to streamline the notation, on the left hand sides of the above equations we have denoted the volume form by $\eta$, suppressing its 2-form basis. The variations with Lin constraints required for the Euler-Poincar\'e equations are then given by,
\begin{equation}\label{eqn:EP-vars}
\begin{aligned}
    \delta u  &=\partial_{t} \xi- \operatorname{ad}_{u} \xi 
    \,,\\
    \delta \eta \,d^2x &=-\mathcal{L}_{\xi} (\eta\,d^2x) =-\operatorname{div}(\eta \boldsymbol{\xi}) d^{2} x 
    \,,\\
    \delta B  &= \ad_{\xi}B = -[\xi,B] =-\mathcal{L}_{\xi} B =(-\boldsymbol{\xi} \cdot \nabla \bB+\bB \cdot \nabla \boldsymbol{\xi}) \cdot \nabla
    \,,
\end{aligned}
\end{equation}
for an arbitrary vector field $\xi = \bs{\xi}\cdot\nabla \in \mathfrak{X}(M)$, where $[\cdot ,\cdot]:\mathfrak{X}\times\mathfrak{X}\rightarrow\mathbb{R}$ denotes the Lie bracket of vector fields. The vector field $u\in\mathfrak{X}(M)$ corresponds to a curve $\phi_t\in{\rm Diff M}$ such that $u=\dot{\phi}\phi^{-1}$, where tangent lifted right translation is denoted by concatenation. Furthermore, the advected quantities $\eta d^2x$ and $B$ evolve by right action of this curve in ${\rm Diff}(M)$. The above constrained variations are those induced on the variables in $\mathfrak{X}\ltimes (\mathfrak{X}\otimes\Omega^2)$ by varying the path $\phi_t$ in diffeomorphism group arbitrarily such that the variations are fixed at the endpoints. Within the variational principle, we will often have need to integrate these terms by parts. The following `diamond' notation is introduced to denote this process
\begin{equation}\label{eqn:diamond-def}
    \scp{\lambda\diamond a}{\xi}_{\mathfrak{X}^*\times\mathfrak{X}} = - \scp{\lambda}{\mathcal{L}_{\xi}a}_{V\times V^*} \,,
\end{equation}
for an advected quantity $a \in V^*$ and a vector field $\xi\in\mathfrak{X}$. For the advected quantities described above, the diamond terms may be computed as
\begin{equation}\label{eqn:diamond_B_and_eta}
    \lambda_B\diamond B = - \ad^*_B\lambda_B  \,,\quad\hbox{and}\quad \lambda_\eta \diamond (\eta\,d^2x) = d\lambda_{\eta} \otimes \eta \, d^2x \,,
\end{equation}
where $\lambda_B \in \mathfrak{X}^*(M)$ and $\lambda_{\eta} \in \Omega^0(M)$. To compute these diamond terms, we have integrated by parts, made use of the form of the Lie derivative of a volume form, and exploited the antisymmetry of the adjoint representation as a map $\ad_{\square}\square : \mathfrak{X}\times\mathfrak{X}\rightarrow\mathfrak{X}$. Furthermore, we have introduced the coadjoint representation of vector fields on the space of $1$-form densities, $\ad^*_{\square}\square:\mathfrak{X}\times\mathfrak{X}^*\rightarrow\mathfrak{X}^*$, which is the dual operator of $\ad$ with respect to the duality pairing defined by $L^2$ spatial integration.

The Euler-Poincar\'e motion equation follows by considering Hamilton's Principle with the variations \eqref{eqn:EP-vars} as
\begin{equation*}
\begin{aligned}
    0 = \delta S &= \int_{0}^{T}\scp{\frac{\delta \ell}{\delta u}}{ \partial_{t} \xi-\operatorname{ad}_{u} \xi }+ \scp{\frac{\delta \ell}{\delta \eta}}{-\mathcal{L}_{\xi} (\eta\,d^2x)} + \scp{\frac{\delta \ell}{\delta B}}{\ad_{\xi} B} \,dt
    \\
    &= \int_{0}^{T}\scp{-\left(\p_t+\ad^*_u\right) \frac{\delta \ell}{\delta u} +  \frac{\delta\ell}{\delta\eta}\diamond (\eta\,d^2x) + \frac{\delta\ell}{\delta B}\diamond B}{\xi} \,dt
    \\
    &= \int_{0}^{T} \scp{-\left(\p_t+\ad^*_u\right) \frac{\delta \ell}{\delta u} +  d\frac{\delta\ell}{\delta\eta}\otimes \eta\,d^2x - \ad^*_B\frac{\delta\ell}{\delta B}}{\xi} \,dt\,.
\end{aligned}
\end{equation*}
Having isolated the arbitrary vector field variation $\xi$, the fundamental lemma of the calculus of variations then yields the Euler-Poincar\'e motion 
equation as%
\footnote{In the motion equation \eqref{eqn:EPmot} we have applied the advection equation \eqref{eqn:advetion-eta} to factor out the volume form $\eta\,d^2x$.} 
\begin{equation}
    \left(\p_t+\mcal{L}_u\right) \frac{1}{\eta}\frac{\delta \ell}{\delta u} = d\frac{\delta\ell}{\delta\eta} - \frac{1}{\eta}\ad^*_B\frac{\delta\ell}{\delta B} \,,
\label{eqn:EPmot}
\end{equation}
to be considered together with the auxiliary equations for $\eta$ in  \eqref{eqn:advetion-eta} and $B$ in \eqref{eqn:advection-B}.
In preparation for writing the final form of the RSW-MHD equations, notice that the form of the variational derivative with respect to $B$ and the condition \eqref{eqn:divergence-condition} imply that
\begin{equation}\label{eq:RSWMHD-coord-free-motion}
    -\frac{1}{\eta}\ad^*_B\frac{\delta\ell}{\delta B} = \frac{1}{\eta\mu^2}\left( (\mcal{L}_B B^\flat)\otimes \eta d^2x + B^\flat \otimes \mcal{L}_B(\eta d^2x) \right) = \frac{1}{\mu^2}\mcal{L}_B B^\flat \,,
\end{equation}
since, on the Lie algebra of vector fields, $\ad^*_B \square$ is a Lie derivative with respect to $B\in\mathfrak{X}(M)$. Inserting the form of the remaining variational derivatives in \eqref{eqn:RSW-MHD-vars} into the Euler-Poincar\'e equation \eqref{eqn:EPmot} yields
\begin{equation}
    \left(\p_t+\mcal{L}_u\right) u^\flat + \frac{1}{\mathrm{Ro}}\mcal{L}_u R = d\left( \frac{1}{2}|\bu|^2 - \frac{1}{2\mu^2}|\bB|^2 + \frac{1}{\mathrm{Ro}}\bu\cdot \bs{R} - \frac{1}{\mathrm{Fr}^2}(\eta - h(\bx)) \right) + \frac{1}{\mu^2}\mcal{L}_B B^\flat \,.
\label{eqn: RSWMHD-motion}
\end{equation}
On a two dimensional domain, in vector calculus notation, the Lie derivative of a $1$-form, $R = \bR\cdot d\bx$, with respect to a vector field, $u = \bu\cdot\nabla$, has the following two equivalent forms
\begin{equation}\label{eq:LieXform}
    \mathcal{L}_u R = (\bu\cdot\nabla\bs{R} + R_j\nabla u^j)\cdot d\bx = ((\nabla^{\perp}\cdot\bs{R})\bu^\perp + \nabla(\bu\cdot \bs{R}))\cdot d\bx \,,
\end{equation}
where $(u_1,u_2)^\perp = (-u_2,u_1)$ and there is an implicit sum over the components of the vectors $\bu,\bs{R}$. In vector calculus form, the Euler-Poincar\'e equation \eqref{eqn: RSWMHD-motion} may therefore be written as
\begin{equation}
    \p_t\bu + \bu\cdot\nabla\bu + \frac{1}{\mathrm{Ro}}(\nabla^{\perp}\cdot\bs{R})\bu^\perp  = - \frac{1}{\mathrm{Fr}^2}\nabla(\eta - h(\bx)) 
    + \frac{1}{\mu^2}\bB\cdot\nabla\bB \,. \label{eq: SWMHD B}
\end{equation}
In two dimensions, $\bs{R}$ is a vector potential for the Coriolis parameter $f$, and these are related to each other by $2\bs{R} = f\bx^\perp$. Therefore, $\nabla^\perp\cdot\bs{R} = f$ and the Euler-Poincar\'e motion equation in \eqref{eqn: RSWMHD-motion} is equivalent to \eqref{eqn:RSW-MHD-motion-VC-nondim}.

\begin{remark}
For a closed material loop of fluid, $c(u)$, moving with the flow of $u$, the Kelvin-Noether theorem for the rotating shallow water MHD equations may be written  as
\begin{equation}
\begin{aligned}
    \frac{d}{dt} \oint_{c(u)} \left( u^\flat + \frac{1}{\mathrm{Ro}}R \right) &=\oint_{c(u)}\left(\p_t+\mathcal{L}_u\right) \left( u^\flat + \frac{1}{\mathrm{Ro}}R \right) = \oint_{c(u)}\left(d\frac{\delta\ell}{\delta\eta} -\frac{1}{\mu^2}\mcal{L}_B B^\flat \right) = -\oint_{c(u)} \frac{1}{\mu^2}\mathcal{L}_B B^\flat  
    \\
    &=\oint_{c(u)} \frac{1}{\mu^2}(\bB \cdot \nabla \bB )\cdot d\bx = \oint_{c(u)} \frac{1}{\mu^2}((\nabla^\perp \cdot \bB)\bB^\perp)\cdot d\bx
    \,,
\end{aligned}
\end{equation}
where from the second to the first line we have converted to a vector calculus notation and the final equality is due to the fact that $\bB\cdot\nabla\bB$ and $(\nabla^\perp \cdot \bB)\bB^\perp$ differ by a gradient. The $B$-field therefore generates circulation unless it is a potential flow. This is in contrast to the three dimensional magnetic fluids, such as Euler's equation coupled to a magnetic field, for which the $B$-field generates circulation unless $\bB$ and its curl are colinear. The analogue between these two cases follows from the fact that, for $\bB_3 = (\bB,0)$ and $\nabla_3 = (\nabla,\p_z)$, we have that $\nabla_3\times\bB_3 = (\nabla^\perp\cdot\bB)\wh{z}$. Hence, $(\nabla^\perp\cdot\bB)\bB^\perp = - \bB_3 \times (\nabla_3\times\bB_3)$. Notice that in the two dimensional case, whereby the dynamics is occurring on an embedded plane within the three dimensional domain, it is not possible for $\bB_3$ and its curl to be co-linear.
\end{remark}
\paragraph{Stream function version of Lagrangian}
As demonstrated in Remark \ref{rmk:divergence-condition}, the expression $\operatorname{div}(\eta \bB)=0$ is advected by the RSW-MHD fluid flow. Thus, one may define a stream function $\psi \in \Omega^0(M)$,
\begin{equation}
    \eta \bB = \nabla^\perp \psi \,, \label{eq:psi def}
\end{equation}
which is advected by the flow $u$ as a scalar quantity. That is,  
\[
    0 = \p_t \psi + \mcal{L}_u \psi = \p_t \psi + \bs{u}\cdot \nabla \psi
\,.\]

Thus, we may rewrite the action \eqref{eqn:RSW-MHD-action} in terms of the equivalent Lagrangian $\wt{\ell} = \wt{\ell}(u,\eta, \psi)$ to have,
\begin{align}
\begin{split}
    S &= \int_0^T \wt{\ell}(u,\eta, \psi)\,dt \\
    &= \int_0^T \int_M \left(\frac{1}{2}|\bs{u}|^{2}+\frac{1}{\mathrm{Ro}} \bs{u} \cdot \boldsymbol{R}(\boldsymbol{x})
    -\frac{1}{2\mu^2}\left|\frac{1}{\eta}\nabla^\perp\psi\right|^{2}
    -\frac{1}{2 \mathrm{Fr}^{2}}(\eta-2 {h}(\boldsymbol{x}))\right) \eta\,d^2x\,, 
\end{split}
    \label{eq:RSW-MHD-action 2}
\end{align}
together with modified constrained Euler-Poincar\'e variations \eqref{eqn:EP-vars} to include the variation of $\psi$,
\begin{align*}
\delta u  &=\partial_{t} \xi- \operatorname{ad}_{u} \xi 
= \big(\partial_{t}  \bs{\xi} + \bs{u}\cdot \nabla  \bs{\xi}  
- \bs{\xi}\cdot \nabla \bs{u}\big)\cdot \nabla
    \,,\\
    \delta \eta \,d^2x &=-\mathcal{L}_{\xi} (\eta\,d^2x) =-\operatorname{div}(\eta \boldsymbol{\xi}) d^{2} x 
    \,,\\
    \delta \psi &= -\mcal{L}_\xi \psi = -\bs{\xi}\cdot \nabla \psi \,,
\end{align*}
for an arbitrary vector field $\xi = \bs{\xi}\cdot \nabla \in \mathfrak{X}(M)$. The Euler-Poincar\'e equation derived from the stationary condition of \eqref{eq:RSW-MHD-action 2} is
\begin{align}
    \left(\p_t+\mcal{L}_u\right) \frac{1}{\eta}\frac{\delta \wt\ell}{\delta u} = d\frac{\delta \wt\ell}{\delta\eta} - \frac{1}{\eta}\frac{\delta \wt\ell}{\delta \psi} d \psi\,,
\label{eqn: psi var}
\end{align}
which are to be considered with the advection equation for $\eta d^2x$ and $\psi$. The variational derivatives of $\wt{\ell}$ in \eqref{eqn: psi var} are given by
\begin{align}
\begin{split}
    &\frac{\delta \wt{\ell}}{\delta u} = \frac{\delta \ell}{\delta u} \in \Omega^1(M)\otimes \Omega^2(M)\,,\quad \frac{\delta \wt{\ell}}{\delta \psi} = \frac{1}{\mu^2}\nabla^\perp\cdot \left(\frac{1}{\eta} \nabla^\perp \psi\right)\otimes d^2x \, \in  \Omega^2(M) \,, \\
    &\frac{\delta \wt{\ell}}{\delta \eta} = \left(\frac{1}{2}|\bs{u}|^{2}+\frac{1}{\mathrm{Ro}} \bs{u} \cdot \boldsymbol{R}(\boldsymbol{x}) + \frac{1}{2\mu^2}\left|\frac{1}{\eta} \nabla^\perp \psi\right|^{2}-\frac{1}{\mathrm{Fr}^{2}}(\eta-{h}(\boldsymbol{x}))\right) \in \Omega^0(M)\,,
\end{split}
\end{align}
Upon inserting these variational derivatives of $\wt{\ell}$ into the Euler-Poincar\'e motion equation, one finds
\begin{equation}
\begin{aligned}
    \left(\p_t+\mcal{L}_u\right) u^\flat + \frac{1}{\mathrm{Ro}}\mcal{L}_u R &= -\frac{1}{\mu^2\eta}\nabla^\perp\cdot\left(\frac{1}{\eta}\nabla^\perp\psi\right)d\psi
    \\
    &\qquad + d\left(\frac{|\bu|^2}{2} + \frac{1}{\mathrm{Ro}}\bu\cdot\bR + \frac{|\nabla^\perp\psi|^2}{2\eta^2\mu^2} - \frac{1}{\mathrm{Fr}^2}(\eta - h(\bx)) \right) \,.
\end{aligned}
\label{eq:RSWMHD-motion}
\end{equation}
As the terms involving the shallow water variables $u,R,\eta$ are the same as before, the momentum equation in vector calculus notation reads as
\begin{equation}
    \p_t\bu + \bu\cdot\nabla\bu + \frac{1}{\mathrm{Ro}}f\bu^\perp = -\frac{1}{\mathrm{Fr}^2}\nabla(\eta - h(\bx)) - \frac{1}{\eta\mu^2}\nabla^\perp\cdot\left(\frac{1}{\eta}\nabla^\perp\psi\right)\nabla\psi + \frac{1}{\mu^2}\nabla\left( \frac{|\nabla^\perp\psi|^2}{2\eta^2} \right) \,.
\label{eqn: psi mot}
\end{equation}
Recalling the definition of $\psi$ in \eqref{eq:psi def}, one has that $\nabla\psi = - \eta\bB^\perp$, and hence the final two terms may be written as
\begin{equation}
    (\nabla^\perp\cdot\bB)\bB^\perp + \nabla\left(\frac{|\bB|^2}{2} \right) = \bB\cdot\nabla\bB \,.
\end{equation}
Consequently, the momentum equation in \eqref{eqn: psi mot} is equivalent to both \eqref{eqn:RSW-MHD-motion-VC} and the Euler-Poincar\'e equation \eqref{eq: SWMHD B}.

\subsection{Hamiltonian structure of RSW-MHD equations}\label{Ham-RSW-MHD-sec}
In this section, the Hamiltonian Lie-Poisson form of the RSW-MHD equations is derived via a Legendre transformation. The Lie-Poisson structure of these equations is interesting (and may be helpful) in studying a variety of astrophysical phenomena, including the dynamics of gravity waves and Alfv\'en waves \citep{zeitlin2024lagrangian,petrosyan2020shallow} in plasmas such as the solar tachocline \citep{gilman2000magnetohydrodynamic,miesch2005large}. 

The Legendre transformation of the Lagrangian corresponding to the action \eqref{eqn:RSW-MHD-action} yields the Hamiltonian in terms of the magnetic vector field, $B=\bB \cdot \nabla$, the columnar volume, $\eta\, d^2x$, and a momentum variable defined as
\begin{equation}\label{eqn:moment map}
    m = \bs{m}\cdot d\bs{x}\otimes d^2x =: \frac{\delta\ell}{\delta u} = u^\flat\otimes \eta\,d^2x + \frac{1}{\mathrm{Ro}}R\otimes \eta\,d^2x \,.
\end{equation}
Explicitly, the Hamiltonian is
\begin{equation}\label{RSWMHD-Ham}
    H(m,\eta,B) = \scp{m}{u} - \ell(u,\eta,B) = \int \frac{1}{2\eta}\left|\bs{m} - \frac{\eta\bs{R}}{\mathrm{Ro}}\right|^2 + \frac{\eta}{2\mu^2}|\bs{B}|^2 + \frac{1}{2\mathrm{Fr}^2}(\eta-2h(\bs{x}))\eta \,d^2x \,. 
\end{equation}
It should be remarked that the momentum variable, $m$, is a $1$-form density and, together with $R\otimes \eta\,d^2x$, appears in the first term of the Hamiltonian. This has been written in a vector form, in terms of $\bs{m}$ and $\bR$, and a function, $\eta$, for clarity. Formally, this term represents the following
\begin{equation*}
    \frac12 \left( \frac{1}{\eta\,d^2x}\left( m - \frac{1}{\mathrm{Ro}}R\otimes \eta\,d^2x \right) \right)^\sharp \intprod \left(\frac{1}{\eta\,d^2x}\left( m - \frac{1}{\mathrm{Ro}}R\otimes \eta\,d^2x \right)  \right) \eta\,d^2x = \frac12(u \intprod u^\flat ) \eta\,d^2x \,.
\end{equation*}
Varying this term in $m$ produces the vector field $u$. Variations of the Hamiltonian are therefore obtained as
\begin{equation}
    \frac{\delta H}{\delta m} = u \,,\quad \frac{\delta H}{\delta B} = \frac{1}{\mu^2} B^\flat \otimes \eta\,d^2x \,,\quad\hbox{and}\quad \frac{\delta H}{\delta \eta} = -\frac{1}{2}|\bu|^2 - \frac{1}{\mathrm{Ro}}\bu\cdot\bR + \frac{1}{2\mu^2}|\bB|^2 + \frac{1}{\mathrm{Fr}^2}(\eta - h(\bx)) \,.
\end{equation}
The semidirect-product Lie-Poisson equations are then given in standard form by
\begin{equation}
\partial_{t}\left(\begin{array}{c}
m \\
B \\
\eta\,d^2x
\end{array}\right) =-\left[\begin{array}{ccc}
\operatorname{ad}_{\square}^{*} m & -\operatorname{ad}_{B}^{*} \square & \eta \,d(\square) \otimes d^2x \\
-\operatorname{ad}_{\square} B & 0 & 0 \\
\mathcal{L}_{\square} (\eta\,d^2x) & 0 & 0
\end{array}\right]\left(\begin{array}{c}
\delta H / \delta m \\
\delta H / \delta B \\
\delta H / \delta \eta
\end{array}\right) \,.\label{eqn:RSWMHD PB}
\end{equation}
The matrix operator here denotes the coadjoint action of the whole semidirect product algebra $\mathfrak{X}(M)\ltimes (\Omega^0(M)\oplus\mathfrak{X}^*(M))$ on its dual, and corresponds to the Lie-Poisson bracket associated with the standard Lie bracket on the space $\mathfrak{X}(M)\ltimes (\Omega^0(M)\oplus\mathfrak{X}^*(M))$. A short calculation verifies that these equations are equivalent to those obtained from the action \eqref{eqn:RSW-MHD-action}.

\paragraph{A stream function for the magnetic field.}
We may write the Hamiltonian in terms of the function $\psi$, rather than the vector field $B$. This Hamiltonian is
\begin{equation}
    H(m,\eta,B) = \wt{H}(m,\eta,\psi) = \int \frac{1}{2\eta}\left|\bs{m} - \frac{\eta\bs{R}}{\mathrm{Ro}}\right|^2 + \frac{1}{2\eta\mu^2}|\nabla^\perp \psi|^2 + \frac{1}{2\mathrm{Fr}^2}(\eta-2h(\bs{x}))\eta \,d^2x \,.
\end{equation}
The variations are computed as
\begin{equation}\label{eqn:variational-derivatives-psi}
    \frac{\delta \wt{H}}{\delta m} = u \,,\quad \frac{\delta \wt{H}}{\delta \psi} = -\frac{1}{\mu^2}\nabla^\perp\cdot\left(\frac{1}{\eta}\nabla^\perp\psi\right)\,d^2x \,,\quad\hbox{and}\quad \frac{\delta \wt{H}}{\delta \eta} = -\frac{1}{2}|\bu|^2 - \frac{1}{\mathrm{Ro}}\bu\cdot\bR - \frac{|\nabla^\perp\psi|^2}{2\eta^2\mu^2} + \frac{1}{\mathrm{Fr}^2}(\eta - h(\bx)) \,.
\end{equation}
Since $\eta$ and $B$ are advected quantities, it follows that $\psi \in \Omega^0(M)$ must be an advected $0$-form (function). The Lie derivative of a function $\psi$ with respect to a vector field $u=\bu\cdot\nabla\in\mathfrak{X}$ takes the form $\mathcal{L}_{u}\psi = \bu\cdot\nabla\psi$. Hence, for functions, the diamond operator defined in equation \eqref{eqn:diamond-def} is computed as
\begin{equation}\label{eqn:diamond-psi}
    \lambda_{\psi}\diamond \psi = -\lambda_{\psi}d\psi \,,
\end{equation}
where $\lambda_{\psi}\in\Omega^2(M)$ is a volume form. We define the following semidirect product Lie-Poisson system
\begin{equation}
    \p_t \left(\begin{array}{c}m \\ \psi \\ \eta\,d^2x\end{array}\right)=-\left[\begin{array}{ccc}\operatorname{ad}_{\square}^{*} m & \square \diamond \psi & \square \diamond (\eta\,d^2x) \\ \mathcal{L}_{\square} \psi & 0 & 0 \\ \mathcal{L}_{\square} (\eta d^2x) & 0 & 0\end{array}\right]\left(\begin{array}{c} \delta\wt{H} / \delta m \\ \delta\wt{H} / \delta \psi \\ \delta\wt{H} / \delta \eta \end{array}\right) \,. \label{eq:psi LP bracket}
\end{equation}
This equation is the Lie-Poisson equation corresponding to the Lie-Poisson bracket associated with the standard semidirect product Lie bracket on $\mathfrak{X}(M)\ltimes (\Omega^2(M)\oplus\Omega^0(M))$.
\begin{proposition}\label{prop:RSWMHS casimirs}
    The shallow water MHD equations possesses the following Casimir invariances,
    \begin{align}\label{eq:RSWMHD casimirs}
        C(q, \psi, \eta) = \int \eta \phi(\psi) + \eta q \varphi(\psi) d^2x
    \end{align}
    for arbitrary smooth functions $\phi$ and $\varphi$. The quantity $q$ denotes  
    the potential vorticity (PV) associated with the flow, which is defined as
    \begin{align}\label{eq:qv def}
    q = \frac{1}{\eta}\big(\nabla^\perp\cdot \bu - \frac{f}{Ro}\big)
    \,.
    \end{align}
\end{proposition}
\begin{proof}
    One can transform the Lie-Poisson bracket \eqref{eq:psi LP bracket} to the potential vorticity bracket for the thermal rotating shallow water (TRSW) equation. As in \cite{holm1989lyapunov}, we use the following transformation $(m, \psi, \eta) \rightarrow (u, \psi, \eta)$, whose Jacobian is given by
    \begin{align*}
        J = \begin{pmatrix}
            \frac{1}{\eta} & 0 & -\frac{m}{\eta^2} \\
            0 & 1 & 0 \\
            0 & 0 & 1 
        \end{pmatrix}\,, \qquad 
        J^T = \begin{pmatrix}
            \frac{1}{\eta} & 0 & 0 \\
            0 & 1 & 0 \\
            -\frac{m}{\eta^2} & 0 & 1 
        \end{pmatrix}\,. \qquad 
    \end{align*}
    Let $\mcal{J}_{LP}$ denote the Lie-Poisson bracket appearing in \eqref{eq:psi LP bracket}. In coordinate notation, the potential vorticity bracket $\mcal{J}_{PV}$ may be calculated, as follows, 
    \begin{align*}
        \mcal{J}_{PV} &= J \circ \mcal{J}_{LP} \circ J^T \\
        & = \begin{pmatrix}
            \frac{1}{\eta} & 0 & -\frac{m_i}{\eta^2} \\
            0 & 1 & 0 \\
            0 & 0 & 1 
        \end{pmatrix}
        \begin{pmatrix}
            m_j\p_i\square + \p_j(m_i\square) & -\p_i \psi & \eta\p_i \square \\
            \p_j \psi & 0 &0 \\
            \p_j(\eta \square ) & 0 & 0 
        \end{pmatrix}
        \begin{pmatrix}
            \frac{1}{\eta} & 0 & 0 \\
            0 & 1 & 0 \\
            -\frac{m_j}{\eta^2} & 0 & 1 
        \end{pmatrix}
        \\
        & = \begin{pmatrix}
            \frac{1}{\eta}\p_j(u_i + R_i/Ro) - \frac{1}{\eta}\p_i(u_j + R_j/Ro) & \frac{1}{\eta}\p_i\psi & \p_i \\
            -\frac{1}{\eta}\p_j\psi & 0 & 0 \\
            \p_j & 0 & 0 
        \end{pmatrix} = \begin{pmatrix}
            q\times & \frac{1}{\eta}\p_i\psi & \p_i \\
            -\frac{1}{\eta}\p_j\psi & 0 & 0 \\
            \p_j & 0 & 0 
        \end{pmatrix} \,,
    \end{align*}
    where in the last equality the two dimensional expressions obtain the potential vorticity bracket for the TRSW equations. From the standard literature on the Casimirs for TRSW equations, e.g., \cite{Zeitlin2018, HLP2021}, one obtains the required results.
\end{proof}
\begin{remark}[Equvalence of Poisson brackets]
    Noting that $B = {\eta}^{-1} \nabla^\perp \psi$, we have 
    \begin{align*}
        J_{PV} = \begin{pmatrix}
            0 & -q & -B_2 & \p_1 \\
            q & 0 & B_1 & \p_2 \\
            B_2 & -B_1 & 0 & 0 \\
            \p_1 & \p_2 & 0 & 0 
        \end{pmatrix}
    \end{align*}
    which is the the same Poisson bracket that appears in \cite{dellar2002hamiltonian}. 
\end{remark}

\section{Thermal effects in shallow water magnetohydrodynamics}\label{TRSW-MHD-sec}

The Thermal Rotating Shallow Water MHD (TRSW--MHD) model is an extension of the RSW--MHD model to include horizontal gradients of buoyancy \citep{dellar2003common}. The TRSW-MHD equations modify the RSW--MHD equations to include the variable buoyancy $b(\bx,t) = \rho(\bx,t))/\bar{\rho}$, where  $\rho$ is the (time and space dependent) mass density and $\bar{\rho}$ is the uniform reference mass density. The non-dimensional TRSW--MHD equation for the fluid velocity $\bu$ is given by
\begin{align}
    \p_t\bu + \bu\cdot\nabla\bu + \frac{1}{\mathrm{Ro}}(\nabla^{\perp}\cdot\bs{R})\bu^\perp  = - \frac{1}{\mathrm{Fr}^2}(1+\mathfrak{s}b)\nabla(\eta - h(\bx)) - \frac{\mathfrak{s}}{2\mathrm{Fr}^2}\eta\nabla b + \frac{1}{\mu^2}\bB\cdot\nabla\bB \,, \label{eq:TRSWMHD-VC}    
\end{align}
where the non-dimensional number $\mathfrak{s}$ is the stratification parameter. Three advection relations now hold as auxiliary equations for the TRSW--MHD velocity equation \eqref{eq:TRSWMHD-VC}. Namely, 
\begin{align}
    \p_t\eta +\nabla \cdot(\eta \bs{u}) &= 0
    \label{eqn:advection-eta-TVC}
    \,,\\
    \p_t\bB + \bs{u} \cdot \nabla \bB 
    - \bB \cdot \nabla \bs{u}&=0
    \label{eqn:advection-B-TVC}
     \,,\\
    \p_t b + \bs{u} \cdot \nabla b &= 0
    \label{eqn:advection-gamma-TVC}
     \,.
\end{align}
The additional auxiliary equation for the positive scalar function $b$ still preserves the divergence condition \eqref{eqn:divergence-condition}, $\nabla \cdot(\eta \bB)=0$, which therefore can still be regarded as a non-dynamical constraint on the initial values of $\eta$ and $\bB$ which does not involve the scalar advected quantity $b$. 

\paragraph{Lagrangian formulation.} 
The TRSW--MHD equations \eqref{eq:TRSWMHD-VC}--\eqref{eqn:advection-gamma-TVC} can be derived as Euler-Poincar\'e equations with advected quantities by considering the following action functional
\begin{equation}\label{eqn:TRSW-MHD-action}
\begin{aligned}
    S &= \int_{0}^{T} \ell(u, \eta, B,b) \,dt 
    \\
    &= \int_{0}^{T} \int_{M}\left(\frac{1}{2}|\bs{u}|^{2}+\frac{1}{\mathrm{Ro}} \bs{u} \cdot \boldsymbol{R}(\boldsymbol{x})
    -\frac{1}{2\mu^2}|\bB|^{2}\right) \eta 
    -\frac{1}{2\mathrm{Fr}^{2}}(1+\mathfrak{s}b)\left(\eta^2- 2\eta{h}(\boldsymbol{x})\right) \,d^2x \, dt
    \,.
\end{aligned}
\end{equation}
Compared to the RSW--MHD Lagrangian \eqref{eqn:RSW-MHD-action}, the TRSW--MHD Lagrangian \eqref{eqn:TRSW-MHD-action} is the same except the factor $1/\mathrm{Fr}^2$ is replaced by $(1+\mathfrak{s}b)/\mathrm{Fr}^2$. The variational derivatives of the action integral \eqref{eqn:TRSW-MHD-action} are evaluated as
\begin{align*}
    \begin{split}
        \frac{1}{\eta} \frac{\delta \ell}{\delta u} & =\left(\bs{u}+\frac{1}{\mathrm{Ro}} \boldsymbol{R}(\boldsymbol{x})\right) \cdot \mathbf{d} \boldsymbol{x}=: \boldsymbol{V}(\boldsymbol{x}, t) \cdot \mathrm{d} \boldsymbol{x}=: V^{\flat} \in \Omega^1(M) 
        \,,\\
        \frac{\delta \ell}{\delta \eta} & =\left(\frac{1}{2}|\bs{u}|^{2}+\frac{1}{\mathrm{Ro}} \bs{u} \cdot \boldsymbol{R}(\boldsymbol{x})-\frac{1}{2\mu^2}|\bB|^{2}-\frac{1}{\mathrm{Fr}^{2}}(1+\mathfrak{s}b)(\eta-{h}(\boldsymbol{x}))\right)=: \beta(\boldsymbol{x}, t) \in \Omega^0(M) 
        \,,\\
        \frac{\delta \ell}{\delta B} & = \frac{1}{\mu^2}\bB \cdot \mathrm{d} \boldsymbol{x} \otimes \eta d^{2} x=: \frac{1}{\mu^2}B^{\flat} \otimes \eta\, d^{2}x \in \mathfrak{X}^*(M)
        \,,\\
        \frac{\delta \ell}{\delta b} & = -\left(\frac{\mathfrak{s}}{2\mathrm{Fr}^2}\eta(\eta - 2h(\bx))\right)\otimes d^2x =: -\gamma(\boldsymbol{x}, t)\otimes d^2x \in \Omega^2(M) 
        \,.
    \end{split}
\end{align*}
The Euler-Poincar\'e motion equation follows by considering $\delta S = 0$ with the constrained variations \eqref{eqn:EP-vars} as well as $\delta b = -\mathcal{L}_{\xi} b$ where $\xi \in \mathfrak{X}(M)$ is the same arbitrary vector field appearing in \eqref{eqn:EP-vars}. The form of the constrained variations follow from the advection equations \eqref{eqn:advection-eta-TVC}, \eqref{eqn:advection-B-TVC} and \eqref{eqn:advection-gamma-TVC}, which can each be written in a coordinate-free geometric form as
\begin{align}
    (\p_t + \mathcal{L}_u)(\eta\,d^2x) &= 0
    \label{eqn:advection-eta-TA}
    \,,\\
    (\p_t + \mathcal{L}_u)B &= 0
    \label{eqn:advection-TB}
     \,,\\
    (\p_t + \mathcal{L}_u)b &= 0
    \label{eqn:advection-TC}
   \,,
\end{align}
and as before, the condition \eqref{eqn:divergence-condition} may be expressed as $\mathcal{L}_B(\eta d^2x) = 0$. Inserting the constrained variations into Hamilton's principle for the TRSW-MHD action yields
\begin{equation*}
\begin{aligned}
    0 = \delta S &= \int_{0}^{T}\scp{\frac{\delta \ell}{\delta u}}{ \partial_{t} \xi-\operatorname{ad}_{u} \xi }+ \scp{\frac{\delta \ell}{\delta \eta}}{-\mathcal{L}_{\xi} (\eta\,d^2x)} + \scp{\frac{\delta \ell}{\delta B}}{\ad_{\xi} B} + \scp{\frac{\delta \ell}{\delta b}}{-\mcal{L}_{\xi} b} \,dt
    \\
    &= \int_{0}^{T}\scp{-\left(\p_t+\ad^*_u\right) \frac{\delta \ell}{\delta u} +  \frac{\delta\ell}{\delta\eta}\diamond (\eta\,d^2x) + \frac{\delta\ell}{\delta B}\diamond B
        + \frac{\delta\ell}{\delta b}\diamond b }{\xi}
 \,dt
    \\
    &= \int_{0}^{T} \scp{-\left(\p_t+\ad^*_u\right) \frac{\delta \ell}{\delta u} 
    +  d\frac{\delta\ell}{\delta\eta}\otimes \eta\,d^2x 
    - \ad^*_B\frac{\delta\ell}{\delta B}
    - \frac{\delta\ell}{\delta b} d b}{\xi} \,dt\,.
\end{aligned}
\end{equation*}
The fundamental lemma of calculus of variations then yields the Euler-Poincar\'e equation of motion
\begin{equation}\label{TRSW-EPmotion}
    \left(\p_t+\mcal{L}_u\right) \frac{1}{\eta}\frac{\delta \ell}{\delta u} = d\frac{\delta\ell}{\delta\eta} - \frac{1}{\eta}\ad^*_B\frac{\delta\ell}{\delta B} 
    - \frac{1}{\eta}\frac{\delta\ell}{\delta b}  d b \,,
\end{equation}
which is to be considered along with the auxiliary advection equations \eqref{eqn:advection-eta-TA}, \eqref{eqn:advection-TB} and \eqref{eqn:advection-TC}, where we have divided the motion equation through by the volume form $\eta\,d^2x$. Through a similar calculation as for the RSW-MHD case above, one may show the equivalence between the vector calculus version of the TRSW-MHD momentum equation \eqref{eq:TRSWMHD-VC} with the Euler-Poincar\'e equation \eqref{TRSW-EPmotion}.

\begin{remark}
For a closed material loop of fluid, $c(u)$, moving with the flow of $u$, the Kelvin-Noether theorem for the TRSW-MHD equations \eqref{TRSW-EPmotion} is
\begin{equation}\label{TRSW-KNthm}
    \frac{d}{dt} \oint_{c(u)} \left( u^\flat + \frac{1}{\mathrm{Ro}}R^\flat \right) 
    =\oint_{c(u)} \frac{1}{\mu^2}(\bB \cdot \nabla \bB )\cdot d\bx 
    +  \oint_{c(u)} \frac{\mathfrak{s}}{2\mathrm{Fr}^2}\left(\eta - 2{h}(\bx)\right) \nabla b \cdot d\bx \,.
\end{equation}
The TRSW-MHD Kelvin--Noether theorem implies that thermal gradients (i.e., $\nabla b$) and magnetic forces can create circulation. Moreover, with the appropriate scaling of $\mu$ and $\mathrm{Fr}$, the  magnetic and gravitational forces can play independent roles on equal footing and both be involved in the balance with the Coriolis force.
\end{remark}

% \todo[inline]{DH: The KN theorem in \eqref{TRSW-KNthm} instructs us to discover the relative importance of the thermal and magnetic forces in terms of dimension-free numbers. The dimension-free ratios are $B_0^2 / H U^2 $ and $1/\mathrm{Fr}^2 = gH/U^2$. Thus, my guess is that the relative amplitude of the two potential energy terms is  $B_0^2/(gH^2)$, which is independent of the velocity, as it should have been. So, what is MagnetoSolarThermal balance for equilibrium solutions? \\
% \centerline{What is the value of $(B_0^2/H)/(gH)= v^2_{Alfven}/v^2_{SW}$?}
% I think there are two very different regimes: 
% $B_0^2/(gH^2)=O(1)$ and $B_0^2/(gH^2)=O(1/{\color{red}\mathrm{Ro}})>>1$ 
% with $\mathrm{Ro}/\mathrm{Fr}^2=O(1)$.
% }

\paragraph{Hamiltonian formulation.} 
The Hamiltonian formulation of TRSW-MHD may be obtained in a similar way to the Hamiltonian formulation of RSW-MHD system in Section \ref{Ham-RSW-MHD-sec}. Through a Legendre transform, the TRSW-MHD Hamiltonian is calculated from the Lagrangian as
\begin{equation}
\label{TRSW-MHD-Ham}
\begin{split}
    \mathfrak{h}(m,\eta,B,b) &= \scp{m}{u} - \ell(u,\eta,B,b) \\ 
    &= \int \frac{1}{2\eta}\left|\bs{m} - \frac{\eta\bs{R}}{\mathrm{Ro}}\right|^2 
    + \frac{\eta}{2\mu^2}|\bs{B}|^2 
    + \frac{1+\mathfrak{s}b}{2\mathrm{Fr}^2}\left(\eta^2- 2\eta {h}(\boldsymbol{x})\right) \,d^2x \,, 
\end{split}    
\end{equation}
where $m = \frac{\delta \ell}{\delta u} = \bs{m}\cdot d\bs{x}\otimes d^2x = \left(u^\flat + R^\flat/\mathrm{Ro}\right)\otimes \eta \, d^2x$ is the same momentum map as in the RSW-MHD case \eqref{eqn:moment map}. The variational derivatives of $\mathfrak{h}$ can found as
\begin{align}
    \begin{split}
        &\frac{\delta \mathfrak{h}}{\delta m} = u \,,\quad \frac{\delta \mathfrak{h}}{\delta B} = \frac{1}{\mu^2} B^\flat \otimes \eta\,d^2x \,, \quad \frac{\delta \mathfrak{h}}{\delta b} = \frac{\mathfrak{s}}{2\mathrm{Fr}^2}\left(\eta^2 - 2\eta h(\bx)\right)\,,\\
        &\quad\hbox{and}\quad \frac{\delta \mathfrak{h}}{\delta \eta} = -\frac{1}{2}|\bu|^2 - \frac{1}{\mathrm{Ro}}\bu\cdot\bR + \frac{1}{2\mu^2}|\bB|^2 + \frac{1+\mathfrak{s}b}{\mathrm{Fr}^2}(\eta - h(\bx)) \,,
    \end{split}
\end{align}
such that the TRSW-MHD system \eqref{eqn:advection-eta-TA}-\eqref{TRSW-EPmotion} can be assembled into the following Lie-Poisson Hamiltonian system,
\begin{equation}
    \p_t \left(\begin{array}{c}m \\ B \\ \eta\,d^2x \\ b \end{array}\right)
    =-\left[\begin{array}{cccc}
    \operatorname{ad}_{\square}^{*} m & -\ad^*_B \square & \square \diamond (\eta\,d^2x)  & \square \diamond  b 
    \\ -\ad_\square B & 0 & 0 & 0 
    \\ \mathcal{L}_{\square} (\eta d^2x) & 0 & 0 & 0 
    \\ \mathcal{L}_{\square} b & 0 & 0 & 0     
    \end{array}\right]
    \left(\begin{array}{c} \delta\mathfrak{h} / \delta m 
    \\ \delta\mathfrak{h} / \delta B \\ \delta\mathfrak{h} / \delta \eta \\  \delta\mathfrak{h} / \delta b
    \end{array}\right) \,. \label{eq:TRSW-MHD LP B bracket}
\end{equation}
The matrix in equation \eqref{eq:TRSW-MHD LP B bracket} is the coadjoint action of $\mathfrak{X}(M)\ltimes (\mathfrak{X}^*(M)\oplus\Omega^0(M) \oplus\Omega^2(M))$ on its dual. 
\begin{proposition}
    The TRSW-MHD equations \eqref{eqn:advection-eta-TA}-\eqref{TRSW-EPmotion} possess the following Casimir invariances,
    \begin{align}
        C(q, \psi, \eta,b) = \int \eta \big(\phi(\psi) + \Phi(b) \big) 
        + \eta q \big(\varphi(\psi) + \Gamma(b)\big)d^2x\,.
    \end{align}
    for arbitrary smooth functions $\phi$, $\varphi$, $\Phi$ and $\Gamma$, where 
    the potential vorticity (PV) associated with the flow is defined by equation \eqref{eq:qv def}
\end{proposition}
\begin{proof}
The proofs of these two propositions for the TRSW-MHD equations run parallel to the corresponding proofs for the RSW-MHD equations.
\end{proof}

\begin{remark}[TRSW-MHD with magnetic stream functions]
    As with the RSW-MHD system, the TRSW-MHD equations can be written using the magnetic stream function $\psi$. In this case, we have a Hamiltonian given by
    \begin{align*}
        \wt{\mathfrak{h}}(m,\eta,\psi,b) = \int \frac{1}{2\eta}\left|\bs{m} - \frac{\eta\bs{R}}{\mathrm{Ro}}\right|^2 + \frac{1}{2\eta\mu^2}|\nabla^\perp \psi|^2 + \frac{1+\mathfrak{s}b}{2\mathrm{Fr}^2}\left(\eta^2- 2\eta {h}(\boldsymbol{x})\right) \,d^2x\,, 
    \end{align*}
    together with the Lie-Poisson system
    \begin{equation*}
    \p_t \left(\begin{array}{c}m \\ \psi \\ \eta\,d^2x \\ b \end{array}\right)
    =-\left[\begin{array}{cccc}
    \operatorname{ad}_{\square}^{*} m & \square \diamond \psi & \square \diamond (\eta\,d^2x)  & \square \diamond  b 
    \\ \mathcal{L}_{\square} \psi & 0 & 0 & 0 
    \\ \mathcal{L}_{\square} (\eta d^2x) & 0 & 0 & 0 
    \\ \mathcal{L}_{\square} b & 0 & 0 & 0     
    \end{array}\right]
    \left(\begin{array}{c} \delta\wt{\mathfrak{h}} / \delta m 
    \\ \delta\wt{\mathfrak{h}} / \delta \psi \\ \delta\wt{\mathfrak{h}} / \delta \eta \\  \delta\wt{\mathfrak{h}} / \delta b
    \end{array}\right) \,, \label{eq:TRSW-MHD LP bracket}
\end{equation*}
where the matrix in equation \eqref{eq:TRSW-MHD LP bracket} defines the Lie-Poisson bracket corresponding to the standard semidirect product Lie-Poisson bracket on the dual of this Lie algebra $\mathfrak{X}(M)\ltimes (\Omega^2(M)\oplus\Omega^0(M) \oplus\Omega^2(M))$. 
\end{remark}

% \section{Thermal Quasi-Magnetostrophy (TQ-M): QG for TRSW-MHD}

% \section{TRSW Hall MHD}

\section{Stochastic rotating shallow water magnetohydrodynamics}\label{Stoch-RSW-MHD}
In this section we consider two complementary approaches for introduucing structure preserving stochastic perturbations to the RSW-MHD equations based on its Lie-Poisson Hamiltonian structure. They are the Stochastic Advection by Lie Transport (SALT) and the Stochastic Forcing by Lie Transport (SFLT) approaches. 
The SALT approach, which introduces stochastic transport noise, as well as stochastic potential energy \citep{holm2015variational, ST2023}, preserves the property of coadjoint orbit motion of the deterministic dynamics as well as the Casimir invariants. The SFLT approach introduces stochastic forcing, as well as stochastic material entrainment effects \citep{HH2021a}, and preserves the energy conservation property of the deterministic dynamics.
For a variational derivation of the SALT and SFLT type perturbations applied to Hall magnetohydrodynamics, see \cite{HHS2024}.

\subsection{Stochastic Advection by Lie Transport (SALT).}\label{sec:SALT}
The general form of the SALT perturbations can be found by considering a family of Hamiltonians defined on the Lie co-algebra $\mathfrak{X}(M)^*\ltimes(\Omega^2(M)\oplus \mathfrak{X}(M))$, denoted by $\{\wt{H}_i(m,\eta,B)\}_{i=1}^N$, as well as a family of i.i.d driving Brownian motions $\{W^i_t\}_{i=1}^N$. 
The Poisson structure \eqref{eqn:RSWMHD PB} of the SALT structure-preserving stochastic RSW-MHD equations is 
\begin{equation}
\begin{aligned}
    \diff \left(\begin{array}{c}m \\ B \\ \eta\,d^2x\end{array}\right)&=-\left[\begin{array}{ccc}
\operatorname{ad}_{\square}^{*} m & -\operatorname{ad}_{B}^{*} \square & \eta \,d(\square) \otimes d^2x \\
-\operatorname{ad}_{\square} B & 0 & 0 \\
\mathcal{L}_{\square} (\eta\,d^2x) & 0 & 0
\end{array}\right]\left(\begin{array}{c}
\delta H / \delta m \\
\delta H / \delta B \\
\delta H / \delta \eta
\end{array}\right) \,\diff t
    \\
    &\qquad - \sum_{i=1}^N \left[\begin{array}{ccc}
\operatorname{ad}_{\square}^{*} m & -\operatorname{ad}_{B}^{*} \square & \eta \,d(\square) \otimes d^2x \\
-\operatorname{ad}_{\square} B & 0 & 0 \\
\mathcal{L}_{\square} (\eta\,d^2x) & 0 & 0
\end{array}\right]\left(\begin{array}{c}
\delta H_i / \delta m \\
\delta H_i / \delta B \\
\delta H_i / \delta \eta
\end{array}\right) \circ \diff W_t^i\,,
\end{aligned}\label{eqn:SALT RSWMHD PB}
\end{equation}
where $\circ$ denotes Fisk-Stratonovich integration. Expanding out the Poisson bracket to individual components of the Lie co-algebra, SALT RSW-MHD advection equations are modified to include stochastic transport, 
\begin{align}
    &\diff\, (\eta\,d^2x) + \mcal{L}_u (\eta\,d^2x) \,\diff t + \mcal{L}_{\frac{\delta H_i}{\delta m}} (\eta\,d^2x) \circ \diff W^i_t = 0\,,\label{eq:SALT RSWMHD LP eta} \\
    &\diff B + \mcal{L}_u B \,\diff t + \mcal{L}_{\frac{\delta H_i}{\delta m}} B \circ \diff W^i_t = 0\,,\label{eq:SALT RSWMHD LP B}
\end{align}
and the stochastic momentum equation is modified as
\begin{align}
    \begin{split}
        \diff m + \ad^*_u m \,\diff t + \ad^*_{\frac{\delta H_i}{\delta m}} m \circ \diff W^i_t &= -\eta d\left(\frac{\delta H}{\delta \eta}\,\diff t + \frac{\delta H_i}{\delta \eta}\circ \diff W^i_t\right)\otimes d^2x \\ 
        & \qquad + \ad^*_B\left(\frac{\delta H}{\delta B}\,\diff t + \frac{\delta H_i}{\delta B}\circ \diff W^i_t\right)\,.
    \end{split}\label{eq:SALT RSWMHD LP m}
\end{align}
Under the stochastic advection of $\eta$ and $B$, the weighted incompressibility condition \eqref{eqn:divergence-condition} of the $B$-field is preserved. That is, when $\nabla\cdot(\eta \bB)\,d^2x = \mcal{L}_B (\eta d^2x)$ is initially zero, it remains so in the subsequent motion. This can be verified via a direct calculation similar to that in Remark \ref{rmk:divergence-condition} to find 
\begin{align*}
    \diff \left(\mcal{L}_B(\eta\,d^2x)\right) + \mcal{L}_u \left(\mcal{L}_B(\eta\,d^2x)\right)\diff t + \mcal{L}_{\frac{\delta H_i}{\delta m}} \left(\mcal{L}_B(\eta\,d^2x)\right)\circ \diff W^i_t =0\,.
\end{align*}

For the concrete choice of the Hamiltonian perturbations ${H}_i(m,\eta, B)$ where ${H}_i$ are linear in $(m,\eta, B)$, i.e., 
\begin{equation}
	H_i(m,\eta, B) = \int_M \bm \cdot \bxi_i + \bB\cdot \boldsymbol{D}_i + \eta \gamma_i \,d^2x\,,\quad i = 1,\ldots, N \,,
\end{equation}
for some prescribed $\xi_i(\bx, t) \in \mathfrak{X}(M)$, $D_i(\bx, t) \in \mathfrak{X}^*(M)$ and $\gamma_i(\bx, t) \in \Omega^0(M)$, the SALT RSW-MHD equations given by equations \eqref{eq:SALT RSWMHD LP eta}-\eqref{eq:SALT RSWMHD LP m} can be expressed component-wise in vector calculus notation as follows 
\begin{align}
	&\diff \eta + \nabla\cdot(\eta\,\bu)\,\diff t + \nabla\cdot(\eta\,\bxi_i) \circ \diff W_t^i = 0 
	\,,\\
	&\diff\bB + \left(\bs{u} \cdot \nabla \bB -\bB \cdot \nabla \bs{u}\right)\,\diff t + \left(\bxi_i \cdot \nabla \bB -\bB \cdot \nabla \bxi_i\right) \circ \diff W_t^i = 0 \,,
\end{align}
and 
\begin{align}
    \begin{split}
        &\diff \bu + (\bu\diff t + \bxi_i\circ \diff W_t^i)\cdot\nabla\bu + \frac{f}{\mathrm{Ro}}(\bu^\perp \,\diff t + \bxi_i^\perp\circ \diff W_t^i) + \left(\frac{1}{\mathrm{Fr}^2}\nabla(\eta - h(\bx)) - \frac{1}{\mu^2}\bB\cdot\nabla\bB \right)\,\diff t \\
        & \qquad =  - \left( u_j\nabla \xi^j_i +\nabla\left(\bxi_i\cdot\frac{\bR}{\mathrm{Ro}}\right) + \nabla \gamma_i - \frac{1}{\eta}\left(\bB\cdot \nabla\boldsymbol{D}_i + (D_i)_j \nabla B^j + \boldsymbol{D}_i \nabla\cdot \bB\right) \right)\circ \diff W_t^i \,.
    \end{split}\label{eq:SALT RSWMHD A}
\end{align}
In \eqref{eq:SALT RSWMHD A}, the multi-index $\xi^j_i$ are the components of the perturbation vector field $\xi_i = \bxi_i\cdot \nabla = \xi^j_i \p_j \in \mathfrak{X}(M)$ and the multi-index $(D_i)_j$ are the components of the perturbation to the magnetic field $D_i = \boldsymbol{D}_i \cdot d\bx\otimes d^2x = (D_i)_j dx^j\otimes d^2x$.

To better mimic the form of the deterministic RSW-MHD equations such that the stochastic perturbations can be attributed as the perturbations to the kinetic energy, gravitational and magnetic potential energies, one can make following particular choice of $\wt{H}_i$ as
\begin{equation}
	H_i(m,\eta, B) = \int_M \frac{\bm}{\eta}\cdot \bxi_i + \eta\bB\cdot \boldsymbol{D}_i + \eta \gamma_i \,d^2x\,,\quad i = 1,\ldots, N \,.
\end{equation}
The variational derivatives of $H_i$ are given by
\begin{align*}
    \begin{split}
        \frac{\delta H_i}{\delta \bm} = \frac{\bxi_i}{\eta} =: \wt{\bxi}_i\,,\quad \frac{\delta H_i}{\delta \bB} = \eta\boldsymbol{D}_i \,,\quad \frac{\delta H_i}{\delta \eta} = -\frac{1}{\eta^2}\bm\cdot\bxi_i +\bB\cdot\boldsymbol{D}_i = -(\bu+\bR)\cdot\wt{\bxi}_i + \bB\cdot\boldsymbol{D}_i\,.
    \end{split}
\end{align*}
In this case, the momentum equation \eqref{eq:SALT RSWMHD LP m} becomes
\begin{align}
    \begin{split}
        &\diff \bu + \bu \cdot\nabla\bu\,\diff t + \frac{f}{\mathrm{Ro}}(\bu^\perp \,\diff t + \wt{\bxi}_i^\perp\circ \diff W_t^i) +  \left(\nabla^\perp\cdot\bu\right)\wt{\bxi}_i^\perp \\
        & \qquad \qquad \qquad =  \left(-\frac{1}{\mathrm{Fr}^2}\nabla(\eta - h(\bx)) + \frac{1}{\mu^2}\bB\cdot\nabla\bB \right)\,\diff t + \left(-\nabla \gamma_i + \left(\nabla^\perp\cdot \boldsymbol{D}_i\right)\bB^\perp \right)\circ \diff W_t^i \,.
    \end{split} \label{eq:SALT RSWMHD VC}
\end{align}
where we have heavily used the two equivalent forms of the Lie derivative in two dimensions, given in equation \eqref{eq:LieXform}, to transform the variables in the Lie derivative between $\bu\cdot\nabla\bs{v}$ and $(\nabla^{\perp}\cdot\bs{v})\bu^\perp$. In particular, the stochastic terms in the equation in this setting take the form
\begin{align*}
	&\mathcal{L}_{\wt{\xi}_i}(u^\flat + \bR\cdot d\bx) -\nabla\left( (\bu+\bR)\cdot\bs{\zeta}_i \right)\cdot d\bx  = \left((\nabla^\perp\cdot \bs{R})\bs{\zeta}_i^\perp + (\nabla^\perp\cdot\bu)\bs{\zeta}_i^\perp \right) \cdot d\bx \,,\\
    &\mathcal{L}_{B}\boldsymbol{D}_i\cdot d\bx - \nabla\left( \bB\cdot\boldsymbol{D}_i \right)\cdot d\bx  = \left(\left(\nabla^\perp\cdot \boldsymbol{D}_i\right)\bB^\perp \right) \cdot d\bx \,,
\end{align*}
and one sees a stochastic contribution to the $J\times B$ force in the last term. 
% \begin{remark}[Interpretation of stochastic terms in \eqref{eq:SALT RSWMHD VC}]
    
% \end{remark}
\begin{remark}[Conservation of Casimirs]
    The SALT RSW-MHD equations \eqref{eq:SALT RSWMHD LP eta}-\eqref{eq:SALT RSWMHD LP m} possess the same Casimir invariants as the deterministic RSW-MHD equations given by \eqref{eq:RSWMHD casimirs}. This may be seen by transforming the stochastic Poisson bracket \eqref{eqn:SALT RSWMHD PB} into the $(n,\psi,\eta)$ variables to find
    \begin{equation}
    \begin{aligned}
        \diff \left(\begin{array}{c}m \\ \psi \\ \eta\,d^2x\end{array}\right)&=-\left[\begin{array}{ccc}\operatorname{ad}_{\square}^{*} m & \square \diamond \psi & \square \diamond (\eta\,d^2x) \\ \mathcal{L}_{\square} \psi & 0 & 0 \\ \mathcal{L}_{\square} (\eta d^2x) & 0 & 0\end{array}\right]\left(\begin{array}{c}
    \delta \wt{H} / \delta m \\
    \delta \wt{H} / \delta \psi \\
    \delta \wt{H} / \delta \eta
    \end{array}\right) \,\diff t
        \\
        &\qquad - \sum_{i=1}^N \left[\begin{array}{ccc}\operatorname{ad}_{\square}^{*} m & \square \diamond \psi & \square \diamond (\eta\,d^2x) \\ \mathcal{L}_{\square} \psi & 0 & 0 \\ \mathcal{L}_{\square} (\eta d^2x) & 0 & 0\end{array}\right]\left(\begin{array}{c}
    \delta \wt{H}_i / \delta m \\
    \delta \wt{H}_i / \delta \psi \\
    \delta \wt{H}_i / \delta \eta
    \end{array}\right) \circ \diff W_t^i\,,
    \end{aligned}\label{eqn:SALT RSWMHD PB psi}
    \end{equation}
    where the stochastic Hamiltonians are related by $\wt{H}_i(m,\eta,\psi) = H_i(m, \eta, \eta^{-1}\nabla^\perp\psi)$. As the Poisson structure of the deterministic part and stochastic part of dynamics are the same, the variational derivatives of the Casimirs still lie in the kernel of the deterministic RSW-MHD Poisson structure given by \eqref{eq:RSWMHD casimirs}.
\end{remark}

\begin{remark}
For a closed material loop of fluid, $c(\diff \chi_t)$, moving with the flow generated by the stochastic vector field $\diff \chi_t = u\,\diff t + \wt{\xi}_i\circ \diff W^i_t$, the Kelvin-Noether theorem for the SALT RSW-MHD equations \eqref{eq:SALT RSWMHD VC} is given by
\begin{equation}
\begin{aligned}
    \diff \oint_{c(\diff \chi_t)} \left( u^\flat + \frac{1}{\mathrm{Ro}}R \right) &=\oint_{c(\diff \chi_t)}\left(\diff + \mathcal{L}_{u\,\diff t} + \mathcal{L}_{\xi_i\circ \diff W^i_t}\right) \left( u^\flat + \frac{1}{\mathrm{Ro}}R \right) \\
    &= \oint_{c(\diff \chi_t)}\left(-d\frac{\delta H}{\delta\eta}\,\diff t - d \frac{\delta H_i}{\delta \eta} \circ \diff W^i_t + \frac{1}{\eta}\left(\ad^*_B \frac{\delta H}{\delta B} \,\diff t + \ad^*_B \frac{\delta H_i}{\delta B} \circ \diff W^i_t\right) \right) 
    \\
    &=\oint_{c(\diff \chi_t)} \left(\frac{1}{\mu^2}(\nabla^\perp\cdot \bB)\bB^\perp \,\diff t + \left(\nabla^\perp\cdot \boldsymbol{D}_i\right)\bB^\perp \circ \diff W^i_t\right)\cdot d\bx
    \,,
\end{aligned}
\end{equation}
Thus, one finds that the deterministic and stochastic forcing generating the circulation dynamics takes the same form. Namely, both forces take the form  $J\times B$  where the stochastic perturbation of the current density is prescribed as $\nabla^\perp\cdot \boldsymbol{D}_i$. 
\end{remark}
% \todo[inline]{OS: We will need to introduce the $J\times B$ force.\\
% RH: This should be introduced in the deterministic case.\\
% DH: OK, Ruiao, that is done in the paragraph after equation \eqref{eqn:RSW-MHD-motion-VC}.}
% \todo[inline]{DH: We should write the Kelvin--Noether theorem and discuss what structures have been preserved by using SALT. This is particularly interesting because one may compare uncertainties caused by noise in transport, magnetic field, depth, or buoyancy as well as the cross effects of noise in more than one field equation. \dots But this would be for future work. }

\subsection{Stochastic Forcing by Lie Transport (SFLT).}\label{sec:SFLT}
We consider the general form of the SFLT perturbations for the RSW-MHD system using the $(m,\psi,\eta)$ variables. This is because care must be taken when introducing stochasticity in the $B$ dynamics to preserve the weighted incompressibility condition $\nabla\cdot(\eta\bB)$. 

The general form of the SFLT perturbations can be defined by a family of forces defined on the Lie co-algebra $\mathfrak{X}^*(M)\ltimes(\Omega^0(M)\oplus \Omega^2(M))$ denoted by
\begin{align}
    f^m_i(\bx,t) = \bs{f}^m_i(\bx,t) \cdot d\bx \in \mathfrak{X}^*(M)\,,\quad f^\psi_i(\bx,t) \in \Omega^0(M)\,,\quad f^\eta_i(\bx,t)\otimes d^2x \in \Omega^2(M)\,,
\end{align}
as well as a family of i.i.d driving Brownian motions $\{W^i_t\}_{i=1}^N$. Here, the forces can be arbitrary functions of the Lie co-algebra variables $(m,\psi, \eta)$ and they can explicitly depend smoothly on both $\bx$ and $t$. 
Upon using the Poisson structure \eqref{eqn:RSWMHD PB}, the SFLT structure-preserving stochastic RSW-MHD equations become 
\begin{equation}
\begin{aligned}
    \diff \left(\begin{array}{c}m \\ \psi \\ \eta\,d^2x\end{array}\right)&=-\left[\begin{array}{ccc}\operatorname{ad}_{\square}^{*} m & \square \diamond \psi & \square \diamond (\eta\,d^2x) \\ \mathcal{L}_{\square} \psi & 0 & 0 \\ \mathcal{L}_{\square} (\eta d^2x) & 0 & 0\end{array}\right]\left(\begin{array}{c} \delta\wt{H} / \delta m \\ \delta\wt{H} / \delta \psi \\ \delta\wt{H} / \delta \eta \end{array}\right)\,\diff t
    \\
    &\qquad - \sum_{i=1}^N \left[\begin{array}{ccc}\operatorname{ad}_{\square}^{*} f^m_i & \square \diamond f^\psi_i & \square \diamond (\eta\,d^2x) \\ \mathcal{L}_{\square} f^\psi_i & 0 & 0 \\ \mathcal{L}_{\square} (\eta d^2x) & 0 & 0\end{array}\right]\left(\begin{array}{c} \delta\wt{H} / \delta m \\ \delta\wt{H} / \delta \psi \\ \delta\wt{H} / \delta \eta \end{array}\right) \circ \diff W_t^i\,.
\end{aligned}\label{eqn:SFLT RSWMHD PB}
\end{equation}
Expanding out the Poisson bracket into individual components of the Lie co-algebra, the SFLT RSW-MHD equations comprise the following advection equations under stochastic forcing,  
\begin{align}
    &\diff\, (\eta\,d^2x) + \mcal{L}_u (\eta\,d^2x) \,\diff t + \mcal{L}_{u} (f^\eta_i\,d^2x) \circ \diff W^i_t = 0\,,\label{eq:SFLT RSWMHD LP eta} \\
    &\diff \psi + \mcal{L}_u \psi \,\diff t + \mcal{L}_{u} f^\psi_i \circ \diff W^i_t = 0\,,\label{eq:SFLT RSWMHD LP B}
\end{align}
as well as the stochastic momentum equation
\begin{align}
    \begin{split}
        \diff m + \ad^*_u m \,\diff t + \ad^*_u f^m_i \circ \diff W^i_t &= -\eta d\frac{\delta  \wt{H}}{\delta \eta}\otimes d^2x \,\diff t - f^\eta_i d \frac{\delta \wt{H}}{\delta \eta}\otimes d^2x \circ \diff W^i_t \\ 
        & \qquad + \frac{\delta \wt{H}}{\delta \psi}d \psi \,\diff t +\frac{\delta \wt{H}}{\delta \psi}d f^\psi_i \circ \diff W^i_t\,.
    \end{split}\label{eq:SFLT RSWMHD LP m}
\end{align}
In applications, as $\eta$ represents the depth of the fluid and is related to the columnar volume of the flow above a given area element, the perturbations $f^\eta_i$ often are set to zero to preserve the advection property of columnar volume due to its relation to the determinant of the fluid back-to-labels map. In this case, the SFLT RSW-MHD equations given by equations \eqref{eq:SFLT RSWMHD LP eta}-\eqref{eq:SFLT RSWMHD LP m} may be expressed component-wise in vector calculus notation as follows 
\begin{align}
	&\diff \eta + \nabla\cdot(\eta\,\bu)\,\diff t = 0 
	\,,\\
	&\diff\psi + \bs{u} \cdot \nabla \psi \,\diff t + \bs{u} \cdot \nabla f^\psi_i \circ \diff W_t^i = 0
	\,,
\end{align}
and 
\begin{align}
    \begin{split}
        &\diff \bu + \left(\bu \cdot\nabla\bu + \frac{f}{\mathrm{Ro}}\bu^\perp + \frac{1}{\mathrm{Fr}^2}\nabla(\eta - h(\bx)) - \frac{1}{\mu^2}\bB\cdot\nabla\bB \right)\,\diff t \\
        & \qquad = - \frac{1}{\eta}\left(\left(\nabla^\perp\cdot \bs{f}^m_i\right)\bu^\perp + \nabla\left(\bu\cdot\bs{f}^m_i \right) + \bs{f}^m_i\nabla\cdot\bu\right)\circ \diff W^i_t + \frac{1}{\mu^2}\left(\nabla^\perp\cdot\bB\right)\boldsymbol{F}_i^\perp \circ \diff W_t^i \,,
    \end{split}\label{eq:SFLT RSWMHD VC}
\end{align}
where we have defined the perturbation to the magnetic field as $\eta \boldsymbol{F}_i = \nabla^\perp f^\psi_i$. Here, the relationship between $\psi$ and $\bB$ is unchanged from equation \eqref{eq:psi def}, so that the weighted incompressibility condition of $B$ in \eqref{eqn:divergence-condition} still holds in the SFLT RSW-MHD dynamics. From the dynamics of $\eta$ and $\psi$, one may obtain the dynamics of $\bB$ as
\begin{align}
    \diff \bB +  \left(\bs{u} \cdot \nabla \bB -\bB \cdot \nabla \bs{u}\right)\,\diff t + \left(\bu \cdot \nabla \boldsymbol{F}_i - \boldsymbol{F}_i \cdot \nabla \bu\right) \circ \diff W_t^i = \frac{\boldsymbol{F}_i}{\eta}\nabla\cdot(\bu \eta)\circ \diff W^i_t\,,
\end{align}
although this is not equivalent to having introduced the SFLT type perturbation into the RSW-MHD Poisson structure \eqref{eqn:RSWMHD PB}, since the term $\eta^{-1}\boldsymbol{F}_i\nabla\cdot(\bu \eta)$ is required to preserve the condition $\nabla\cdot(\bu\eta)$.

\begin{remark}[Conservation of energy]
    The SFLT RSW-MHD equations \eqref{eq:SFLT RSWMHD LP eta}-\eqref{eq:SFLT RSWMHD LP m} preserve the domain integrated Hamiltonian of the deterministic RSW-MHD equations given by \eqref{TRSW-MHD-Ham}. This result may be obtained by noticing that both the deterministic and stochastic Poisson structures of the SFLT RWS-MHD equations appearing in \eqref{eqn:SFLT RSWMHD PB} are skew-symmetric under the duality pairing of the Lie algebra and its dual. As both Poisson structures are applied to the same Hamiltonian, it follows that $\wt{H}$, $\wt{H}$ is conserved by the SFLT dynamics.
\end{remark}

\begin{remark}
For a closed material loop of fluid, $c(u)$, moving with the flow generated by the semi-martingale vector field $u$, the Kelvin-Noether theorem for the SFLT RSW-MHD equations \eqref{eq:SFLT RSWMHD VC} is given by
\begin{equation}
\begin{aligned}
    \diff \oint_{c(u)} \left( u^\flat + \frac{1}{\mathrm{Ro}}R \right) &=\oint_{c(u)}\left(\diff + \mathcal{L}_{u\,\diff t} \right) \left( u^\flat + \frac{1}{\mathrm{Ro}}R \right) \\
    &= \oint_{c(u)}\left(-d\frac{\delta \wt{H}}{\delta\eta}\,\diff t + \frac{1}{\eta}\left(-\ad^*_u f^m_i + \frac{\delta \wt{H}}{\delta\eta}\nabla \psi \,\diff t + \frac{\delta \wt{H}}{\delta\eta}\nabla f^\psi_i \circ \diff W^i_t\right) \right) 
    \\
    &=\oint_{c(u)} \left(\frac{1}{\mu^2}(\nabla^\perp\cdot \bB)\bB^\perp \,\diff t \,\diff t + \left(\nabla^\perp\cdot \bB\right)\boldsymbol{F}_i^\perp \circ \diff W^i_t\right)\cdot d\bx \\
    & \qquad - \oint_{c(u)} \frac{1}{\eta}\left(\left(\nabla^\perp\cdot \bs{f}^m_i\right)\bu^\perp + \nabla\left(\bu\cdot\bs{f}^m_i \right) + \bs{f}^m_i\nabla\cdot\bu\right)\cdot d\bx \circ \diff W^i_t
    \,.
\end{aligned}
\end{equation}
In the circulation dynamics, one sees that the momentum perturbation $f^m_i$ and the magnetic field perturbation $\boldsymbol{F}_i$ both generate circulation. The forcing due to $f^m_i$ takes the form of Stokes drift and the forcing due to $\boldsymbol{F}_i$ takes the form of the $J\times B$ force defined by the magnetic field perturbation, $\boldsymbol{F}_i$.
% \todo[inline]{OS: We will need to introduce the $J\times B$ force.\\
% DH; This has now been done in the paragraph after equation \eqref{eqn:RSW-MHD-motion-VC}.}
\end{remark}

\section{Conclusion and outlook}\label{Conc-Out-sec} 
The present work has made several distinct contributions. In Section \ref{Intro-RSW-MHD-sec} we investigated the geometric structures of the Rotating Shallow Water Magenetohydrodynamics (RSW-MHD) equations proposed in Gilman \cite{gilman2000magnetohydrodynamic} as a model of solar tachocline dynamics. In Subsection \ref{RSW-sec}, the RSW-MHD equations were rederived via Lie group invariant variational principles using the Euler-Poincar\'e formalism pioneered in \cite{HMR1998}. Two equivalent representations of the prognostic variable for magnetic effects were developed, in terms of either the magnetic field $B$ or the magnetic stream function $\psi$. In Subsection \ref{Ham-RSW-MHD-sec}, the Lie--Poisson Hamiltonian formulation of the RSW-MHD equations was constructed via a Lie group reduced Legendre transform in either the variables involving $B$ or those involving $\psi$. In Subsection \ref{Ham-RSW-MHD-sec}, we also characterized the Casimir functions of the RSW-MHD system and demonstrated the equivalence of the Lie-Poisson bracket derived here to the potential vorticity bracket of the RSW-MHD equations derived previously, e.g., in \cite{dellar2002hamiltonian}.

In Section \ref{TRSW-MHD-sec}, we extended the treatment of RSW-MHD in Section \ref{Intro-RSW-MHD-sec} to include thermal gradients in the variational principle and thereby derive the \emph{Thermal} Rotating Shallow Water MHD (TRSW-MHD) model. 

In Section \ref{Stoch-RSW-MHD}, we constructed stochastic RSW-MHD models using two distinct and complementary types of structure preserving stochastic perturbations. These two methodologies for stochastic perturbation are known within the literature as Stochastic Advection by Lie Transport (SALT) and Stochastic Forcing by Lie Transport (SFLT). The SALT RSW-MHD model, presented in Subsection \ref{sec:SALT} preserves the deterministic Casimir functions and possesses stochastic $J \times B$ forces from which the uncertainty in the \emph{current density} can potentially be calibrated through data. On the other hand, the SFLT RSW-MHD model, presented in Subsection \ref{sec:SFLT} preserves the deterministic energy and also possesses stochastic $J \times B$ forces by which the \emph{magnetic field} can potentially be calibrated through data.

\paragraph{Open problems and future work.}
The present work has assembled a modeling framework for quasi-neutral plasma dynamics and its stochastic perturbations, which we expect may find utility in many other domains. In future work, we plan to use stochastic perturbations of MHD flows to model the effects of unresolved fluctuations. In particular, we plan to quantify uncertainty in stochastic plasma models via ensemble simulations at coarse grid resolutions and compared with either fine grid resolution simulations or highly resolved observational data. Our earlier efforts in developing and applying this procedure in geophysical fluid dynamics (GFD) have already proven its potential utility in plasma physics. For examples of this procedure in GFD, see, e.g., \cite{CCHOS18a} for the SALT approach and \cite{HP2023} for the SFLT approach. 

There are also several open problems that may be addressed in this framework. For example, the stochastic variational methods applied in the present paper may be transferred to a variety of other plasma models, such as multi-fluid plasma dynamics \cite{Holm1987}, as well as Maxwell-Vlasov and other hybrid kinetic/continuum plasma dynamics \cite{CHHM1998, HT2012a}.

The mean field approximation of the SALT and SFLT modelling approaches may also be applied to plasma models. Examples of these mean-field methods are the Lagrangian Averaged SALT approach \cite{DHL2020} and Eulerian Averaged SFLT \cite{HH2021a} approach.

The structure preserving stochastic perturbations applied in the present paper have recently been generalised to go beyond Brownian stochastic flows to include flows on geometric rough paths. See, e.g. \citep{CHLN2022,DHP2023} where the flows on geometric rough paths are introduced in a Lie group invariant variational principle. These structure preserving rough path perturbations can be used to derive models for plasma dynamics on rough paths, where the same uncertainty estimation and data calibration procedures for stochastic plasma dynamics can be applied.

\subsection*{Acknowledgements} 
We are grateful to J. Woodfield, C. Cotter, D. Crisan, T. Diamantakis, and H. Dumpty for several thoughtful suggestions during the course of this work which have improved or clarified the interpretation of its results. 
DH and RH were partially supported during the present work by Office of Naval Research (ONR) grant award N00014-22-1-2082, Stochastic Parameterization of Ocean Turbulence for Observational Networks. DH and OS were partially supported during the present work by European Research Council (ERC) Synergy grant Stochastic Transport in Upper Ocean Dynamics (STUOD) -- DLV-856408. OS acknowledges funding for a research fellowship from Quadrature Climate Foundation, which has partially supported his contribution to this project.

\bibliographystyle{apalike}
\bibliography{main}

\begin{thebibliography}{}

\bibitem[Alonso-Or{\'a}n, 2021]{alonso2021asymptotic}
Alonso-Or{\'a}n, D. (2021).
\newblock Asymptotic shallow models arising in magnetohydrodynamics.
\newblock {\em Water Waves}, 3(2):371--398.

\bibitem[Cendra et~al., 1998]{CHHM1998}
Cendra, H., Holm, D.~D., Hoyle, M. J.~W., and Marsden, J.~E. (1998).
\newblock {The Maxwell–Vlasov equations in Euler–Poincaré form}.
\newblock {\em Journal of Mathematical Physics}, 39:3138--3157.

\bibitem[Cotter et~al., 2019]{CCHOS18a}
Cotter, C., Crisan, D., Holm, D.~D., Pan, W., and Shevchenko, I. (2019).
\newblock {Numerically Modeling Stochastic Lie Transport in Fluid Dynamics}.
\newblock {\em Multiscale Modeling \& Simulation}, 17:192--232.

\bibitem[Crisan et~al., 2022]{CHLN2022}
Crisan, D., Holm, D.~D., Leahy, J.-M., and Nilssen, T. (2022).
\newblock Variational principles for fluid dynamics on rough paths.
\newblock {\em Advances in Mathematics}, 404:108409.

\bibitem[De~Sterck, 2001]{de2001hyperbolic}
De~Sterck, H. (2001).
\newblock Hyperbolic theory of the shallow water magnetohydrodynamics equations.
\newblock {\em Physics of plasmas}, 8(7):3293--3304.

\bibitem[Dellar, 2002]{dellar2002hamiltonian}
Dellar, P.~J. (2002).
\newblock Hamiltonian and symmetric hyperbolic structures of shallow water magnetohydrodynamics.
\newblock {\em Physics of Plasmas}, 9(4):1130--1136.

\bibitem[Dellar, 2003a]{dellar2003common}
Dellar, P.~J. (2003a).
\newblock Common {H}amiltonian structure of the shallow water equations with horizontal temperature gradients and magnetic fields.
\newblock {\em Physics of Fluids}, 15(2):292--297.

\bibitem[Dellar, 2003b]{dellar2003dispersive}
Dellar, P.~J. (2003b).
\newblock Dispersive shallow water magnetohydrodynamics.
\newblock {\em Physics of Plasmas}, 10(3):581--590.

\bibitem[Diamantakis et~al., 2023]{DHP2023}
Diamantakis, T., Holm, D.~D., and Pavliotis, G.~A. (2023).
\newblock {Variational Principles on Geometric Rough Paths and the Lévy Area Correction}.
\newblock {\em SIAM Journal on Applied Dynamical Systems}, 22:1182--1218.

\bibitem[Drivas et~al., 2020]{DHL2020}
Drivas, T.~D., Holm, D.~D., and Leahy, J.-M. (2020).
\newblock Lagrangian averaged stochastic advection by lie transport for fluids.
\newblock {\em Journal of Statistical Physics}, 179:1304--1342.

\bibitem[Gilman, 2000]{gilman2000magnetohydrodynamic}
Gilman, P.~A. (2000).
\newblock Magnetohydrodynamic `shallow water' equations for the solar tachocline.
\newblock {\em The Astrophysical Journal}, 544(1):L79.

\bibitem[Holm, 2015]{holm2015variational}
Holm, D.~D. (2015).
\newblock Variational principles for stochastic fluid dynamics.
\newblock {\em Proceedings of the Royal Society A: Mathematical, Physical and Engineering Sciences}, 471(2176):20140963.

\bibitem[Holm and Hu, 2021]{HH2021a}
Holm, D.~D. and Hu, R. (2021).
\newblock Stochastic effects of waves on currents in the ocean mixed layer.
\newblock {\em Journal of Mathematical Physics}, 62.

\bibitem[Holm et~al., 2024]{HHS2024}
Holm, D.~D., Hu, R., and Street, O.~D. (2024).
\newblock {Deterministic and stochastic geometric mechanics for Hall magnetohydrodynamics}.
\newblock {\em Proceedings of the Royal Society A: Mathematical, Physical and Engineering Sciences}, 480.

\bibitem[Holm and Kupershmidt, 1987]{Holm1987}
Holm, D.~D. and Kupershmidt, B.~A. (1987).
\newblock Superfluid plasmas: Multivelocity nonlinear hydrodynamics of superfluid solutions with charged condensates coupled electromagnetically.
\newblock {\em Physical Review A}, 36:3947--3956.

\bibitem[Holm and Long, 1989]{holm1989lyapunov}
Holm, D.~D. and Long, B. (1989).
\newblock Lyapunov stability of ideal stratified fluid equilibria in hydrostatic balance.
\newblock {\em Nonlinearity}, 2(1):23.

\bibitem[Holm et~al., 2021]{HLP2021}
Holm, D.~D., Luesink, E., and Pan, W. (2021).
\newblock Stochastic mesoscale circulation dynamics in the thermal ocean.
\newblock {\em Physics of Fluids}, 33:046603.

\bibitem[Holm et~al., 1998]{HMR1998}
Holm, D.~D., Marsden, J.~E., and Ratiu, T.~S. (1998).
\newblock The euler–poincaré equations and semidirect products with applications to continuum theories.
\newblock {\em Advances in Mathematics}, 137:1--81.

\bibitem[Holm and Tronci, 2012]{HT2012a}
Holm, D.~D. and Tronci, C. (2012).
\newblock Euler-poincaré formulation of hybrid plasma models.
\newblock {\em Communications in Mathematical Sciences}, 10:191--222.

\bibitem[Hu and Patching, 2023]{HP2023}
Hu, R. and Patching, S. (2023).
\newblock {\em Variational Stochastic Parameterisations and Their Applications to Primitive Equation Models}, pages 135--158.
\newblock Springer International Publishing.

\bibitem[Hughes et~al., 2007]{hughes2007solar}
Hughes, D.~W., Rosner, R., and Weiss, N.~O. (2007).
\newblock {\em The solar tachocline}.
\newblock Cambridge University Press.

\bibitem[Hunter, 2015]{hunter2015waves}
Hunter, S. (2015).
\newblock {\em Waves in shallow water magnetohydrodynamics}.
\newblock PhD thesis, University of Leeds.

\bibitem[Lahaye and Zeitlin, 2022]{LZ2022}
Lahaye, N. and Zeitlin, V. (2022).
\newblock Coherent magnetic modon solutions in quasi-geostrophic shallow water magnetohydrodynamics.
\newblock {\em Journal of Fluid Mechanics}, 941:A15.

\bibitem[Mak et~al., 2016]{mak2016shear}
Mak, J., Griffiths, S.~D., and Hughes, D.~W. (2016).
\newblock Shear flow instabilities in shallow-water magnetohydrodynamics.
\newblock {\em Journal of Fluid Mechanics}, 788:767--796.

\bibitem[Miesch, 2005]{miesch2005large}
Miesch, M.~S. (2005).
\newblock Large-scale dynamics of the convection zone and tachocline.
\newblock {\em Living Reviews in Solar Physics}, 2(1):1.

\bibitem[Petrosyan et~al., 2020]{petrosyan2020shallow}
Petrosyan, A., Klimachkov, D., Fedotova, M., and Zinyakov, T. (2020).
\newblock Shallow water magnetohydrodynamics in plasma astrophysics. waves, turbulence, and zonal flows.
\newblock {\em Atmosphere}, 11(4):314.

\bibitem[Rossmanith, 2003]{rossmanith2003constrained}
Rossmanith, J.~A. (2003).
\newblock {A constrained transport method for the shallow water MHD equations}.
\newblock In {\em Hyperbolic Problems: Theory, Numerics, Applications: Proceedings of the Ninth International Conference on Hyperbolic Problems held in CalTech, Pasadena, March 25--29, 2002}, pages 851--860. Springer.

\bibitem[Schecter et~al., 2001]{schecter2001shallow}
Schecter, D., Boyd, J., and Gilman, P. (2001).
\newblock Shallow-water magnetohydrodynamic waves in the solar tachocline.
\newblock {\em The Astrophysical Journal}, 551(2):L185.

\bibitem[Street and Takao, 2024]{ST2023}
Street, O.~D. and Takao, S. (2024).
\newblock Semimartingale driven mechanics and reduction by symmetry for stochastic and dissipative dynamical systems.

\bibitem[Zeitlin, 2013]{zeitlin2013remarks}
Zeitlin, V. (2013).
\newblock Remarks on rotating shallow-water magnetohydrodynamics.
\newblock {\em Nonlinear Processes in Geophysics}, 20(5):893--898.

\bibitem[Zeitlin, 2018]{Zeitlin2018}
Zeitlin, V. (2018).
\newblock {\em Geophysical Fluid Dynamics: Understanding (almost) everything with rotating shallow water models}.
\newblock Oxford University Press.

\bibitem[Zeitlin, 2024]{zeitlin2024lagrangian}
Zeitlin, V. (2024).
\newblock Lagrangian approach to nonlinear waves in non-dispersive and dispersive rotating shallow water magnetohydrodynamics.
\newblock {\em Journal of Fluid Mechanics}, 983:A42.

\end{thebibliography}

\end{document}